%% file: main.tex
\pdfoutput=1
\documentclass[10pt]{llncs}

\input{preamble.tex}

\input{macros.tex}

\title{A Tableau Method for the Realizability and Synthesis of Reactive Safety Specifications
}

\newcommand{\montseOrcid}{\orcidID{0000-0001-5627-501X}}
\newcommand{\paquiOrcid} {\orcidID{0000-0001-7872-2685}}
\newcommand{\cesarOrcid} {\orcidID{0000-0003-3927-4773}}

\author{Montserrat Hermo\inst{1}\montseOrcid \and
  Paqui Lucio\inst{1}\paquiOrcid \and
   C\'esar S\'anchez\inst{2}\cesarOrcid}
\institute{University of the Basque Country, San Sebasti\'an, Spain
  \and
  IMDEA Software Institute, Madrid, Spain}

\begin{document}

\maketitle


\begin{abstract}
   We introduce a tableau decision method for deciding realizability of
  specifications expressed in a safety fragment of LTL that includes
  bounded future temporal operators.
  Tableau decision procedures for temporal and modal logics have been
  thoroughly studied for satisfiability and for translating temporal
  formulae into equivalent B\"uchi automata, and also for model
  checking, where a specification and system are provided.
  However, to the best of our knowledge no tableau method has been
  studied for the reactive synthesis problem.

  Reactive synthesis starts from a specification where propositional
  variables are split into those controlled by the environment and
  those controlled by the system, and consists on automatically
  producing a system that guarantees the specification for all
  environments.
  Realizability is the decision problem of whether there is one such
  system.


  In this paper we present a method to decide realizability of safety
  specifications, from which we can also extract (i.e. synthesize) a
  correct system (in case the specification is realizable).
  Our method can easily be extended to handle richer domains
  (integers,  etc) and bounds in the temporal operators in
  ways that automata approaches for synthesis cannot.
\end{abstract}

\input{introduction}
\input{safety}
\input{tnf}
\input{tableaux}

\input{examples}

\input{games}
\input{correctness}
\input{conclusions}

\bibliographystyle{plain}
\bibliography{references}
\end{document}

%% file: preamble.tex
\mainmatter
\usepackage{hyperref}
\usepackage[utf8]{inputenc}
\usepackage{amssymb}
\usepackage{amsmath}
\usepackage{xcolor}
\usepackage{tikz}
\usetikzlibrary{automata, positioning, arrows}
\usepackage{float}

\usepackage{ltlfonts}

\usepackage{forest}

\usepackage[T1]{fontenc}
\usepackage{color}
\usepackage{rotating}
\usepackage{ascii}
\usepackage[normalem]{ulem}

\usepackage{tabularx}
\usepackage{mathtools}

\usepackage[ruled,vlined,linesnumbered]{algorithm2e}

\forestset{
  tableaux/.style={
    for tree={
       math content,
    },
    closed/.style={
      label=below:$\otimes$
    },
    open/.style={
     label=below:$\bigcirc$
    },
    textone/.style={
      label=below:{\hspace{9mm}\scriptsize{##1}}
    },
    texttwo/.style={
      label=below:{\scriptsize{##1}}
    },
    tableau/.style={
      label=below:{##1},draw
    },
  },
}
\usepackage{todonotes}
\usepackage{paralist}
\usepackage{xspace}
\usepackage{comment}
\usepackage{cancel}
\usepackage{lipsum}

\def\orcidID#1{\smash{\href{http://orcid.org/#1}{\protect\raisebox{-1.25pt}{\protect\includegraphics{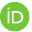}}}}}

%% file: macros.tex
\newcommand{\cF}{{\bf\scriptstyle F}}
\newcommand{\cT}{{\bf\scriptstyle T}}

\newcommand{\semic}{,}

\newcommand{\MX}{\ensuremath{\mathcal{X}}}
\newcommand{\MY}{\ensuremath{\mathcal{Y}}}
\newcommand{\MV}{\ensuremath{\mathcal{V}}}

\newcommand{\tupleof}[1]{\langle#1\rangle}
\newcommand{\KWD}[1]{\ensuremath{\textit{#1}}\xspace}
\newcommand{\trace}{\KWD{trace}}
\newcommand{\move}{\KWD{move}}
\newcommand{\nat}{\mathbb{N}}

\newcommand{\Pos}{P\xspace}
\newcommand{\Into}{\rightarrow}
\newcommand{\G}{\mathcal{G}\xspace}
\newcommand{\A}{\mathcal{A}\xspace}

\newcommand{\Not}{\mathclose{\neg}}
\renewcommand{\And}{\mathrel{\wedge}}
\newcommand{\Or}{\mathrel{\vee}}
\newcommand{\Impl}{\mathrel{\rightarrow}}
\newcommand{\Iff}{\mathrel{\leftrightarrow}}

\newcommand{\Always}{\LTLsquare}
\newcommand{\Next}{\LTLcircle}
\newcommand{\Until}{\,\mathcal{U}\,}
\newcommand{\Eventually}{\LTLdiamonduc}

\newcommand{\Val}{\KWD{\small\sf Val}}

\newcommand{\Tb}{\mathcal{T}}

\newcommand{\Incons}{\KWD{\small\sf Incnst}}
\newcommand{\Cons}{\KWD{\small\sf Cnst}}
\newcommand{\Stt}{\KWD{\small\sf Stt}}

\newcommand{\Clo}{\KWD{\small\sf Clo}}
\newcommand{\Preclo}{\KWD{\small\sf Preclo}}
\newcommand{\SubForm}{\KWD{\small\sf SubFm}}
\newcommand{\Variants}{\KWD{\small\sf Varnt}}
\newcommand{\Ordnf}{\KWD{\small\sf Ordnf}}
\newcommand{\Positions}{\KWD{\small\sf Pos}}

\newcommand{\False}{\KWD{False}}
\newcommand{\True}{\KWD{True}}
\newcommand{\IsOpen}{\KWD{is\_open}}
\newcommand{\Tableau}{\KWD{\textup{\small\sf Tab}}}
\newcommand{\SelectApplicableRuleTo}{\KWD{select\_saturation\_rule}}

\newcommand{\wdash}{\models^{\textit{fin}}}

\newcommand{\depth}{\KWD{\sf depth}}
\newcommand{\ltl}{\KWD{\sf LTL}}
\newcommand{\Mod}{\KWD{\sf Mod}}
\newcommand{\NNF}{\KWD{\sf NNF}}
\newcommand{\DNF}{\KWD{\sf DNF}}
\newcommand{\TNF}{\KWD{\sf TNF}}
\newcommand{\Lit}{\mathcal{L}}
\newcommand{\Fut}{\mathcal{F}}

%% file: introduction.tex
\section{Introduction}
\label{sec:intro}

Linear Temporal Logic (LTL)~\cite{pnueli77temporal} is modal logic for
expressing correctness properties of reactive systems.
Verification is the problem of deciding, given a system $S$ and an LTL
specification $\varphi$, whether $S$ models $\varphi$.
Reactive
synthesis~\cite{pnueli89onthesythesis,pnueli89onthesythesis:b}, first
studied by Pnueli and Rosner in 1989, is the problem of automatically
producing a system $S$ from a given LTL specification $\varphi$ which
guarantees that $S$ models $\varphi$.
In the reactive synthesis problem, the atomic variables are split into
those variables controlled by the environment and the rest, controlled
by the system.

In the last two decades the reactive synthesis problem has been
studied
extensively (see e.g.~\cite{bohy12acacia,ehlers11unbeast,FBS13,finkbeiner14detecting,jobstmann07anzu}).
The approaches can be classified into three categories: (1)
approaches based on games~\cite{buchi69solving}, (2) approaches
that cover a strict fragment of LTL, like GR(1) specifications ~\cite{piterman06synthesis,bloem12synthesis}; (3) bounded
synthesis~\cite{schewe07bounded} explore the problem fixed bound on the size of the system.
In all these cases, the state space of the game arena is either
captured by an automaton or explored explicitly or symbolically.
In this paper we study a deductive alternative: a tableau method for
realizability and synthesis for the class of safety specifications.

Tableaux methods were originally created \cite{Smullyan1963} as very
intuitive deduction procedures for classical propositional and
first-order logic.
A tableau is a tree that performs a symbolic handling of formulas
according to very simple and intuitive rules based on semantics,
model-theory and proof-theory.
Indeed, classical tableaux corresponds to deductive proofs in
Gentzen's sequent calculus.
Tableaux has been evolving for years to decide the satisfiability
problem of many other non-classical logics (modal, multi-valued,
temporal, etc.), in some cases combining with other formal structures,
such as e.g. different kinds of automata.
 
Traditional tableau techniques for satisfiability do not directly work
for realizability.
As far as we know, tableau techniques has not been yet applied for
solving the realizability problem of temporal formulas, beyond its
auxiliary use in automata-based methods \cite{brenguier14abssynthe}.
We present in this paper a tableau based methods for the realizability
of reactive safety specifications.
To illustrate the problem, consider the following three safety
formulas where $e$ is an environment variable $e$ and $s$ is a system
variable:
\begin{eqnarray*}
\psi_1 & = & s \Iff  e\\
\psi_2 & = & s \Iff  \Next e\\
\psi_3 & = & \Next s \Iff  \Next e
\end{eqnarray*}
The safety specifications $\Always\psi_1$ and $\Always\psi_3$ are
realizable.
To see this, consider the strategy where the system simply mimics in
$s$ the value observed in $e$.
On the other hand $\Always\psi_2$ is not realizable, as the system is
required to guess the next value of $e$ which requires clairvoyance.
A temporal tableau for $\Always\psi_1$, first decomposes the formula
into $s\Iff e,\Next\Always(s\Iff e)$ and, second, splits two branches
for the two cases $s,e,\Next\Always(s\Iff e)$ and
$\Not s, \Not e,\Next\Always(s\Iff e)$.
In both nodes there is not more to do than `jumping' to the next
state.
Hence in both branches we get a loop to the root $\Always(s \Iff e)$.
Therefore, each branch represents a model of the initial formula.
Tableau structure with two leaves pointing to the root can be
interpreted as the winning strategy of the system that witnesses the
realizability of the safety specification.
However, following the same steps $\Always\psi_2$ fails to capture the
unrealizability of this specification.
Indeed, the formulas that successively appear in just one branch,
where new-line represents `jump' to next state, are
\[
\begin{array}{l}
\Always(s\Iff \Next e), s\Iff \Next e,\Next\Always(s\Iff \Next e), s,\Next e\\
e, \Always(s\Iff \Next e), s\Iff \Next e, \Next\Always(s\Iff \Next e), s, \Next e\\
e, \Always(s\Iff \Next e)\\
\end{array}
\]
The other branches are similar.
The branch above should be seen as a strategy that, no matter what the
environment does at the start, the system chooses to satisfy $s$,
which in turn forces the environment to satisfy $e$ in the next-step.
Obviously, since the environment could freely choose $\Not e$ this
tableau branch should be closed at node $e, \Always(s\Iff \Next e)$.
A branch closing condition would successfully solve this example:
whenever it is the environment turn, if there is a literal on an
environment variable, the branch is closed (the environment is the
winner by choosing a contradiction).
This closing condition however fails to capture the realizability of
safety specifications such as $\Always\psi_3$.
The tableau for $\Always(\Next s\Iff \Next e)$ would be formed by the
following two closed branches:
\[
\begin{array}{l}
\Always(\Next s\Iff \Next e), \Next s\Iff \Next e,\Next\Always(\Next s\Iff \Next e), \Next s,\Next e\\
s, e, \Always(\Next s\Iff \Next e)\\[3mm]
\Always(\Next s\Iff \Next e), \Next s\Iff \Next e,\Next\Always(\Next s\Iff \Next e), \Next \Not s,\Next \Not e\\
\Not s, \Not e, \Always(\Next s\Iff \Next e)
\end{array}
\]
According to our previous closing condition, the first branch is
closed by $e$ and the second by $\Not e$, whereas the safety
specification $\Always(\Next s\Iff \Next e)$ is realizable.
The problem here is in the splitting of the two cases
$\Next s,\Next e$ and $\Next \Not s,\Next \Not e$.
Intuitively, the environment should be free to choose at the second
state (really, at any state) either $e$ or $\Not e$, but in the second
state of the two branches above the choice is already taken, because
the splitting of the `strict' future formula $\Next s\Iff \Next e$ has
been already made at the first state.
To overcome this problem, we introduce in this paper the {\it terse
  normal form} of formulas which prevents these incorrect splittings
on formulas that reveal future choices too early.
Intuitively, at the second state, our tableau will have just one node:
\begin{equation}
\label{eq-node}
(s \And e) \Or (\Not s \And \Not e), \Always(\Next s\Iff \Next e)    
\end{equation}
This node has two children (one for each choice of the environment):
$$
\begin{array}{l}
e, s,\Next(( s \And  e) \Or (\Not s \And  \Not e)), \Next\Always(\Next s\Iff \Next e)\\[3mm]
\Not e, \Not s, \Next(( s \And  e) \Or ( \Not s \And  \Not e)), \Next\Always(\Next s\Iff \Next e)
\end{array}
$$
Then, `jumping' to the next-state from both nodes produces again node
($\ref{eq-node}$).
This tableau witnesses (see Ex. \ref{ex:intro-example}) that $\Always\psi_3$ is
realizable.
In summary, classical tableau rules for the logical
operators does not provide a correct decision procedure for
realizability.

In this paper we present a tableau method that decides the realizability
problem for a fragment of LTL which includes temporal operators of the form $\Always_{[n,m]}$ and $\Eventually_{[n,m]}$ (for $n,m\in\nat$ such that $n \leq m$).
Although, from the semantic point of view, these operators can be seen as a short-hand for a boolean combinations of formulas using only the "next time" 
operator ($\Next$), the compact notation is effectively exploited in tableau deductions in a more efficient way that prevents their unfoldings.
As an example of such benefits consider the safety formula $\psi_4 = e \Impl \Always_{[0,2^{100}]}s$. A tableau for $\Always\psi_4$ according to the rules we propose, splits two branches for the
two cases $(\neg e \And \Next \Always\psi_4)$ and $(e \And  s \And \Next \Always_{[0,2^{100}-1]}s \And \Next \Always\psi_4)$. The first branch ends up jumping to the next state, which loops to the root $\Always\psi_4$. The second branch jumps to the next state with the formula   $( \Always_{[0,2^{100}-1]}s \And \Always\psi_4)$ and again  two new branches start from here. Each of these branches ends up with a loop to its previous state and the tableau finishes certifying that $\Always\psi_4$ is realizable.

 This example illustrates a crucial difference between automata and tableaux: the tableau ``deduction'', after checking two successive states, is able to decide the realizability of $\Always\psi_4$, whereas automata techniques require an
explicit blasting. 
Throughout the paper we provide more examples remarking this point.
As far as we know, there is no previous temporal tableaux for solving the realizability problem of this logic, so no technique for reactive realizability/synthesis enjoys the benefits of deductive temporal tableaux. Though this paper is mainly dedicated to the realizability problem, our tableaux provides an easy way for synthesis of both kind of certificates: the realizability strategy and the counterexample in the case of unrealizability. Indeed,
the constructed tableau is really a deductive proof of the un/realizability of the input formula.

In summary, the contributions of this paper are:
\begin{compactitem}
\item The introduction of a new normal form, called terse normal form
  that captures precisely, in a logical form, the timely choices of
  the environment and the responses by the system.
\item A tableau method including all the deductive rules to build the
  tableau graph and rules to close the branches, with success and with
  failure. Some of these rules are non-deterministic, which enable
  heuristics for exploring the graph in different forms (depth-first,
  breadth first, etc).
\item Sound and completeness proofs for the tableau method.
\end{compactitem}

\subsubsection*{Related Work.}

Current approaches to reactive
synthesis~\cite{bohy12acacia,ehlers11unbeast,FBS13,finkbeiner14detecting,jobstmann07anzu}
can be classified into three categories: (1) approaches based on
games~\cite{buchi69solving} that translate $\varphi$ into a
deterministic automaton and determine the winner of the game played on
the state graph of the automaton, (2) approaches that cover a strict
fragment of LTL, like GR(1) specifications (the simpler fragment of
general reactivity of rank
1)~\cite{piterman06synthesis,bloem12synthesis}; (3) bounded
synthesis~\cite{schewe07bounded}, which process a set of constraints
that characterizes all correct systems up to fixed bound on the size.
Modern implementations of the the game approach use a symbolic
representation~\cite{jobstmann07anzu} of the arena of the game, or use
decision procedures based on SAT or QBF~\cite{bloem14sat}.
Existing tools for full LTL synthesis, including
Unbeast~\cite{ehlers11unbeast} and Acacia+~\cite{bohy12acacia} are
based on the bounded synthesis idea, where different encoding of the
(decidable) constraint for a given bound have been
proposed~\cite{finkbeiner07smt,schewe07bounded,FBS13} (first-order
logic module finite integer arithmetic), using lazy constraint
generation~\cite{finkbeiner12lazy}, rewriting and modular
solving~\cite{khalimov13towards}, SAT
encodings~\cite{shimakawa15reducing,finkbeiner16bounded}, QBF
encoding~\cite{finkbeiner15bounded}, etc.
Since 2014, the reactive synthesis competition
(SYNTCOMP)~\cite{syntcomp,jacobs17reactive} tries to compare the
performance of synthesis tools against different benchmark problems.

Reactive synthesis for full LTL specifications is
2EXPTIME-complete~\cite{pnueli89onthesythesis}, so synthesis for more
restricted fragments of LTL that exhibit better complexity has been
identified.
For example, synthesis of GR(1) specifications---the simpler fragment
of general reactivity of rank 1---has an efficient polynomial time
symbolic synthesis algorithm~\cite{piterman06synthesis,bloem12synthesis}.
GR(1) synthesis has been used in various application domains and
contexts, including robotics~\cite{kressgazit09temporal}, event-based
behavior models~\cite{dippolito13synhesizing}, etc.
It is easy to see that even though GR(1) covers the safety fragment of
LTL, translating specifications into the language that we consider in this paper
involves at least an exponential blow-up in the worst case.
All the methods listed above correspond to an algorithmic exhaustive
exploration of the game arena.
In contrast, we introduce in this paper a deductive tableau method
for realizability.

  %
  %
  The first tableau method \cite{wolper1985tableau} for the
  satisfiability problem of LTL is not purely tree-shape, instead they
  require the constructions of a graph that should be explored in a
  second pass.
  This inspired a connection with Büchi automata
  \cite{wolper1985tableau,vardi1994reasoning} that is in the origin of
  many automata-based decision procedures \cite{DDMR08c} for LTL
  satisfiability and model-checking problems.
  The use of an auxiliary structure raised two difficulties: One is
  the construction of a big structure and the other is the lost of the
  original correspondence with sequent proofs that could certify the
  result.
  Some alternative ideas
  (e.g. \cite{schwendimann1998new,gore2009optimal}) has been developed
  to explore on-the-fly the graph (or automaton) thorough its
  construction instead in a second pass, and also for constructing
  one-pass tableaux that preserves the correspondence with sequent
  proofs (cf. \cite{GHL09}).

\subsubsection{Outline of the paper.}
The rest of the paper is structured as follows.
Section~\ref{sec:safety} presents the fragment of LTL for safety
specifications.
Section~\ref{sec:tnf} introduces the terse normal form.  which is
allows the realizability tableaux, presented in
Section~\ref{sec:tableaux}.
Several examples are shown in Section~\ref{sec:examples}.
Section~\ref{sec:correctness} includes and proofs of correctness and
finally, Section~\ref{sec:conclusions} concludes.


%% file: safety.tex
\section{Safety Specifications and Games}
\label{sec:safety}

In this section, we first define a sublanguage of \ltl to express
safety properties.
Our language is aimed at industrial specifications of reactive
systems.
We present a method to decide realizability of these specifications in
Section~\ref{sec:tableaux}.
We will introduce safety games and games for safety specifications,
which will serve as method to prove correctness of our approach.
We start by revisiting \ltl.

\subsection{Linear Temporal Logic}
\label{subsect:LTL}

Linear Temporal Logic (\ltl)~\cite{pnueli77temporal}, extends
propositional logic by introducing temporal operators $\Next$ (next)
and $\Until$ (until).
\ltl formulas are interpreted in traces on their set of (occurring) variables $\MV$.
A {\it trace} $\sigma$ is a denumerable sequence of states $\sigma_0,\sigma_1,\sigma_2,\dots$ where each state $\sigma_i$ is a valuation from $\MV$ to $\{0,1\}$. 
A valuation $v$ on a set of variables $\MV$ is a function that assigns to each $p\in\MV$ a boolean value $\{0,1\}$. We denote by $\Val(\MV)$ the set of all valuations on $\MV$.
For any $i\geq 0$, $\sigma^i$ denotes the trace $\sigma_i,\sigma_{i+1},\dots$ and $\sigma^{<i}$ denotes the finite sequence
$\sigma_0,\dots,\sigma_{i-1}$. Therefore, $\sigma = \sigma^{<i}\cdot\sigma^i$. We also use $\sigma^{i..j}$ for the finite sequence between $i$ and $j$ (including $\sigma_i$ but not $\sigma_j$), which can be defined as $(\sigma^i)^{<j}$. 
Given a trace $\sigma$ and an \ltl formula $\phi$ the meaning of $\sigma\models \phi$ is  inductively defined as follows:
\[
\begin{array}{lcl}
\sigma \models p  &\mbox{ iff } & \sigma_0(p) = 1\\
\sigma\models\neg\phi &\text{iff} & \sigma\not\models\phi\\
\sigma\models \phi\And\psi &\text{iff}& \sigma\models\phi \mbox{ and } \sigma\models\psi \\
\sigma \models\Next \phi    &\mbox{ iff }& 
\sigma^{1} \models \phi
\end{array}
\]
\[
\begin{array}{lcl}
 \sigma\models \phi\Until\psi  &\mbox{ iff } &
  \sigma^j \models \psi  \textrm{ for some } j  \textrm{ such that }  0 \leq j \textrm{ and } \\ 
  & & \sigma^i \models \phi  \textrm{ for all } i  \textrm{ such that }  0 \leq i < j.
\end{array}
\]

There are standard abbreviations of constants propositions ($\cT$ for truth and $\cF$ for falsehood), classical connectives (such as $\vee,\Impl$ and $\Iff$) and also of temporal operators such as `eventually': $\Eventually \phi$ for $\cT\Until\phi$ and `always': $\Always\phi$ for
$\Not(\cT\Until\Not\phi)$.

A set of formulas is (syntactically) consistent if and only if it does not contain a formula and its negation.
If $\sigma\models \phi$ then  we say that $\sigma$ models (or is a model of) $\phi$.
We denote $\Mod(\phi)$ to the set of all traces that are models of $\phi$.
A finite set of formulas is intentionally confused with the conjunction of all it members.
Therefore, for a set of formulas $\Phi$, $\sigma\models\Phi$ if and only if $\sigma\models\phi$ for all $\phi\in\Phi$ and $\Mod(\Phi)$ denotes the set of traces that are models of all $\phi\in\Phi$.
A set of formulas $\Phi$ is satisfiable if and only if there exists at least one $\sigma$ such that $\sigma\models\Phi$, otherwise $\Phi$ is unsatisfiable. Two formulas $\phi$ and $\psi$ are logically equivalent, denoted $\phi\equiv\psi$, if and only if $\Mod(\phi) = \Mod(\psi)$.

%
Any sublanguage $\mathcal{G}$ of LTL is {\it safety} whenever for any formula $\phi\in\mathcal{G}$  and for any trace $\sigma$, if $\sigma\not\models\phi$ then there exists some $i > 0$ such that $\sigma^{<i}\cdot\sigma'\not\models\phi$ for any trace $\sigma'$. In words,  any trace with prefix $\sigma^{<i}$ is not a model of $\phi$. In such case, $\sigma^{<i}$ is called a {\it witness of the violation} of $\phi$ (a.k.a. a {\it bad prefix} for $\phi$).

\subsection{Safety Specifications}
\label{subsec:safety}
Our safety specifications are constructed over some set of variables $\MX\cup\MY$ that is partitioned into two subsets: $\MX$ is the set of variables that are controlled by the environment and $\MY$ is the set of variables controlled by the system.
In what follows, each variable in $\MX$ is marked with $e$ as a subscript (e.g. $sensor_e$ or $p_e$). Variables in $\MY$ do not have any subscript.

Boolean formulas are constructed from atoms that could be either propositional variables or equations $x = c$ where $x$ is a variable of an enumerated type $T$ and $c$ is a constant value of type $T$. Compound formulas are constructed
using the classical boolean connectives ($\Not, \And, \Or, \Impl, \Iff$). 
We also use the boolean constants $\cT$ (for truthness) and $\cF$ (for falsehood) as atomic formulas.
More precisely, the grammar for any boolean formula $\eta$ is:
\[
\begin{array}{lcl}
 a & ::= & p \mid x = c \mid \cT \mid \cF\\
\beta & :: = &  {\it a} \mid \Not \beta \mid \beta \And \beta \mid \beta \Or  \beta \mid \beta \Impl \beta \mid \beta \Iff \beta
\end{array}
\]
where  $p$ is a boolean variable, $x$ has enumerated type  $T$ and $c\in T$.
The formulas of the form $a, \Not a$, $x = c$ and $\Not(x = c)$ are called literals. 
%
We consider that the valuations in traces assign to each variable $x$ of enumerated type $T$ a value in $T$, which gives the atom $x = c$ a truth value.

We consider a fragment of \ltl specifications of the form
$\alpha\And\Always\psi$ where $\alpha$ is an initial formula and
$\psi$ is a safety formula.
The formula $\alpha$, called the {\it initial formula}, is a boolean
constraint that captures the initial states
The formula $\Always\psi$ {\it safety constraint} that restricts the
transition relation by means of temporal operators. Here, $\psi$ is a
temporal formula over $\MX\cup\MY$, called the {\it safety formula},
whose grammar is defined as follows:
\[
\begin{array}{lcl}
\eta & ::= & \beta \mid \Not\eta \mid \Next\eta \mid \Always_{I}\eta\mid \Eventually_{I}\eta\mid\eta\Or\eta\mid\eta\And\eta\mid\eta\Impl\eta\mid\eta\Iff\eta  
\end{array}
\]
where $I = [n,m]$ for 
Therefore, safety formulas are constructed adding to the boolean connectives the temporal operators $\Eventually_{I}$ and $\Always_{I}$. The interpretation of this operators in traces are
\[
\begin{array}{lcl}
 \sigma\models \Always_{[n,m]}\eta  &\mbox{ iff } &
\sigma^j \models \eta  \textrm{ for all } j  \textrm{ such that }  n \leq j \leq m\\ 
 \sigma \models \Eventually_{[n,m]}\eta  &\mbox{ iff } &
 \textrm{there exists } j  \textrm{ such that }  n \leq j \leq m \textrm{ such that }  \sigma^j \models \eta 
\end{array}
\]
Note that $\Eventually_{I}$ (resp. $\Always_{I}$) can be expressed as a disjunction (resp. conjunction) of formulas that only use $\Next$ as temporal operator.
A trace $\sigma$ models $\alpha\And\Always\psi$ whenever
$ \sigma_0(\alpha) = 1$ and $\sigma^k\models \psi$ for all $k\geq0$.

It is easy to see that any safety formula is logically equivalent to a formula in {\it Negation Normal Form} (\NNF), i. e. a formula in the following grammar:
\[
\begin{array}{lcl}
 \ell & ::= & p \mid \Not p \mid x = c \mid \Not (x = c) \mid \cT \mid \cF\\
 \eta & ::= & \ell \mid \Next \eta \mid \Always_{I} \eta\mid \Eventually_{I} \eta\mid \eta\Or \eta\mid \eta\And\eta  
\end{array}
\]
It suffices to use, besides
classical logical equivalences on boolean connectives, the following equivalences on temporal connectives:
$$
\begin{array}{lll}
\Not \Next \eta  \equiv  \Next \Not \eta&
\Not \Eventually_I \eta  \equiv  \Always_I \Not \eta&
\Not \Always_I \eta  \equiv  \Eventually_I \Not \eta
\end{array}
$$
In what follows, we assume that safety formulas are translated to their equivalent formulas in \NNF, $\ell$ stands for a literal, and, for $i\in\nat$, $\Next^i$ abbreviates a sequence of  operators $\Next$ of length $i$.

The {\it (temporal) depth} of a safety formula is the maximum number
of nested $\Next$ operators in the formula, where $\Always_I$ and
$\Eventually_I$ are interpreted as being expressed in terms of $\Next$. Formally,
\begin{equation*}
\depth(\eta) = \left\{
        \begin{array}{ll}
            0 & \quad \eta \text{ is a literal}\\[1mm]
            max\{\depth(\eta_1), \depth(\eta_2) \} & \quad 
            \eta = \eta_1 \Or  \eta_2  \mbox{ or }   \eta = \eta_1 \And \eta_2\\[1mm]
            1 + \depth(\eta_1) & \quad \eta = \Next \eta_1 \\[1mm]
            m + \depth(\eta_1) & \quad 
            \eta = \Always_{[n,m]} \eta_1  \mbox{ or }   \eta = \Eventually_{[n,m]} \eta_1
        \end{array}
    \right.
\end{equation*}
%

%
\noindent The following result states that the language of our safety formulas is a safety sublanguage of \ltl.
\begin{lemma}
\label{prop-witness1}
For any safety formula $\eta$ and any trace $\sigma$, $\sigma\not\models\eta$ if and only if exists $d \leq\depth(\eta)$ such that $\sigma^{<d+1}$ is a witness of the violation of $\eta$.
\end{lemma}
\begin{proof}
The backward direction holds by the definition of witness of the violation of $\eta$.
The proof of the forward direction is done by structural induction on $\eta$.\qed
\end{proof}

\begin{corollary}
  \label{cor:witness}
For any safety specification $\varphi = \alpha\And\Always\psi$ and any trace $\sigma$, $\sigma\not\models\varphi$ if and only if one of the following holds:
\begin{enumerate}[(i)]
    \item $\sigma_0$ is a witness of the violation of $\alpha\And\psi$, or
    \item there exists some $i$ and some $d \leq\depth(\psi)$ such that $\sigma^{i..d+1}$ is a witness of the violation of $\psi$.
\end{enumerate}
\end{corollary}

We next introduce the concept of {\it witness of a safety specification}, which allows us to work with a finite semantics for safety specifications.

\begin{definition}
\label{def-witness}

We define the relation $\wdash$ on finite sequences $\lambda = \lambda_0 \cdots\lambda_{d-1}$ where $d\geq 1$ and formulas as follows:
\[
\begin{array}{lcl}
\lambda\wdash \ell &\text{iff}& \lambda_0(\ell)=1 \\
\lambda\wdash \eta_1\And\eta_2 &\text{iff}& \lambda\wdash\eta_1 \mbox{ and } \lambda\wdash\eta_2 \\
\lambda\wdash \eta_1\Or\eta_2 &\text{iff}& \lambda\wdash\eta_1 \mbox{ or } \lambda\wdash\eta_2\\
\lambda\wdash \Next\eta &\text{iff}& \text{if } d>1 \text{ then } \lambda^{1..d}\wdash\eta \\
\lambda\wdash \Always_{[n,m]}\eta & \text{iff} & 
\lambda^j\wdash\eta \text{ for all } n\leq j \leq \min(m,d)\\
\lambda\wdash \Eventually_{[n,m]}\eta &\text{iff}& \text{if }n \leq m < d \text{ then } \lambda^j\wdash\eta \text{ for some $n\leq j \leq m$} \\
\end{array}
\]

A {\it pre-witness of a safety specification} $\alpha\And\Always\psi$ is a finite sequence
$\lambda = \lambda_0 \cdots\lambda_{d-1}$ 
such that $\lambda_0\wdash\alpha\And\psi$ and $\lambda^{<i}\wdash\psi$ for every $1\leq i\leq d-1$.
A {\it witness of a safety specification} $\alpha\And\Always\psi$ is a pre-witness $\lambda_0, \cdots, \lambda_{d-1}$ such that for some $0\leq j < d$ it holds that $\lambda_0\cdots(\lambda_j\cdots\lambda_{d-1})^\omega\models\alpha\And\Always\psi$.
\end{definition}

To construct traces, we usually deal with valuations that satisfies a set of literals.
\begin{definition}
Let $\Delta$ be any set of formulas, we denote by $\Val_\Delta(\MV)$ the set of all valuations $v\in\Val(\MV)$ such that $v(x) = 1$ for every boolean variable $x\in\Delta$,  $v(x) = 0$ for every boolean variable $\Not x\in\Delta$, $v(x) = c$ for every $x$ of enumerated type such that $x = c\in\Delta$, and
$v(x) \not = c$  for every $x$ of enumerated type such that $\Not(x = c)\in\Delta$.
\end{definition}
Note that for each variable $x$ that does not occur in $\Delta$, there are many $v\in\Val_\Delta(\MV)$ with different values for $v(x)$ and also that 
if $\Delta$ is a set of literals then $\lambda_0\wdash\Delta$ if and only if $\lambda_0 \in \Val_\Delta(\MV)$.

In safety specifications, variables come from two disjoint sets of variables $\MX$ and $\MY$. Given $v\in\Val(\MX)$ and $w\in\Val(\MY)$, we denote by $v+w$
the valuation in $z\in\Val(\MX\cup\MY)$ such that $z(p) =v(p)$ if $z\in\MX$ and $z(p) =w(p)$ if $z\in\MY$. This notation is extended to pairs of finite traces $\lambda$ on $\MX$ and $\lambda'$ on $\MY$ of the same length $d$, i.e. $\lambda + \lambda'$ denotes the trace $(\lambda_0+\lambda'_0)\cdots(\lambda_{d-1}+\lambda'_{d-1})$.

\subsection{Safety Games, Games for Safety Specifications}
\label{subsec:games}

We now give an introduction to safety games and games for safety
specifications.

\subsubsection{Safety Games.}
%
A \emph{safety game} is a game played by two players: Eve (the
environment) denoted by $E$, and Sally (the system), denoted by $S$.
The set of \emph{positions} of the game is partitioned into those
position where Eve plays and those where Sally plays.
A \emph{move} consists on the player that owns the current position
choosing a successor position.
A \emph{play} is a infinite sequence of moves starting from some
positions within a predetermined set of initial positions.
The outcome of a play of the game is determined as follows.
Eve wins if some move during the play reaches some bad positions for a
predetermined subset of bad positions.
Otherwise, Sally wins.
In other words, Sally wins a play if she avoids bad positions at all
times during the play.

Formally, an {\em arena} is a tuple $\A:\tupleof{I,\Pos,\Pos_E,\Pos_S,T}$
\begin{itemize}
\item $\Pos$ is the set of positions, partitioned into
  $\Pos=\Pos_E\cup \Pos_S$ and $\Pos_E\cap \Pos_S=\emptyset$
\item $I\subseteq P$ is the set of initial positions;
\item $T\subseteq (\Pos\times \Pos)$ is the set of moves; and
\end{itemize}
A play $\pi:v_0v_1v_2\ldots$ is an infinite sequence of positions
$v_i\in P$.
A game equips an arena with a winning condition $W\subseteq P^w$. A
play $\pi$ is winning for Sally if $\pi\in{}W$.
Safety games are defined by a set of bad states $B\subseteq{}P$ that
encodes the following winning condition
$W_B=\{ \pi \;|\; \text{for all $i$ },\pi(i)\notin B\}$
Therefore, Sally wins a safety game $\mathcal{G}:(\A,B)$ in those
infinite plays where bad positions $B$ are avoided.

We assume that every state has a successor so we do not have to deal
with finite plays.
This does not limit the class of games because we can extend every
arena with a sink state that is either good (resp. bad) if the halting
state is good (resp. bad), generating an equivalent arena that admits
only infinite plays.
Given a play and a natural number $i$, we use $\pi(i)$ to denote the
$i$-th state of $\pi$, given two natural numbers $i<j$, $\pi[i,j]$
denotes the finite sequence $\pi(i),\ldots,\pi(j)$, and $\pi(i..)$
denotes the infinite sequence $\pi(i)\pi(i+1)\ldots$
A play is initial if $\pi(0)\in I$.

A \emph{strategy} $\rho_S$ for Sally is a map
$\rho_S:\Pos^*\times \Pos_S\Into \Pos$, such that for every sequence
of positions $s\in \Pos^*$ followed by a position of Sally
$v\in \Pos_S$, $T(v,\rho_S(s,v))$ is a legal move in the game.
A strategy is memory-less if for all $s_1,s_2\in \Pos^*$ and
$v\in \Pos_S$, $\rho_S(s_1,v)=\rho_S(s_2,v)$, that is, if the move
that the strategy represents is determined solely by the last
position.
In this case we represent $\rho_S$ as a function
$\rho_S:\Pos_S\Into \Pos$ and the move $\rho_S(v)$ is determined by
$v$.
Similarly, a strategy for Eve is a map
$\rho_E:\Pos^*\times \Pos_E\Into \Pos$, such that for every
$s\in \Pos^*$ and $v\in \Pos_E$, $T(v,\rho_E(s,v))$ is a legal move
for Eve in the game.
A play $\pi$ is played according to a strategy $\rho_S$ if for every
$i$, if $\pi(i)\in P_S$ then $\pi(i+1)=\rho_S(\pi[0,i],\pi(i))$ (if
$\rho_S$ is memory-less then $\pi(i+1)=\rho_S(\pi(i))$).
The definition is analogous for $\rho_E$.
A strategy $\rho_S$ of Sally is winning if every initial play $\pi$
played according to $\rho_S$ is winning for Sally.
It is well-known that safety games are determined (either a game has a
winning strategy for Sally or it has a winning strategy for Eve), and
that safety games are memory-less determined (either a game has a
memory-less winning strategy for Sally or it has a memory-less winning
strategy for Eve).

\subsubsection{Games for Safety Specifications.}
%
Consider a safety specification $\varphi$ over $\MX$ and $\MY$.
Given a set $S$, $S^*$ denotes the set of finite strings over $S$ and
$S^k$ the set of strings over $S$ of length $k$.
The build an arena $\A(\MX,\MY)$ of a safety game as follows:
\begin{itemize}
\item $P_E=\{\Val(\MX)^k \times \Val(\MY)^k \mid k\in \nat\}$.
We use
$P_E^k=\{ (\overline{x},\overline{y}) |\;|x|=|y|=k\}\subseteq P_E$.
\item $P_S=\{\Val(\MX)^{k+1} \times \Val(\MY)^k  \mid k\in\nat\}$.
We use $P_S^{k+1}$ for the set
$\{ (\overline{x},\overline{y}) |\;|x|=k+1\text{ and } |y|=k\}\subseteq P_S$.
\item $T$ contains two types of edges $T=T_E\cup T_S$ defined as follows for each $k\in\nat$:
\begin{itemize}
\item $T_E\subseteq(P_E^k,P_S^{k+1})$ such that $((\overline{x},\overline{y}),(\overline{x}\cdot v,\overline{y}))\in T_E$ if $v\in\Val(\MX)$.
\item $T_S\subseteq(P_S^{k+1},P_E^{k+1})$ such that $((\overline{x}\cdot v,\overline{y}),(\overline{x}\cdot v,\overline{y}\cdot w))\in T_S$ if $w\in\Val(\MY)$.
\end{itemize}
\item $I = \{(\epsilon,\epsilon)\}.$ 
\end{itemize}
%
%

The arena described above considers all sets of possibles moves by
Eve and Sally.
Note that Eve and Sally alternate playing.
Given a position $p\in P_E\setminus I$ of the form
$(\overline{x}\cdot v,\overline{y}\cdot w)$ we use $\move(p)=(v+w)$
for the valuation of the variables of $\MX\cup\MY$ according to $v$
and $w$.
Given a play $\pi$ we use $\trace(\pi)$ for the trace $\sigma$ such
that $\sigma(i)=\move(\pi(2i+1))$.
Note that $\trace(\pi)$ corresponds to the sequence of valuations that
Eve and Sally pick.
The arena is essentially an infinite tree that records the valuations
chosen in turns by Eve and Sally.
Given a specification $\varphi$ (for example, a safety specification),
the set of winning plays is $W_\varphi=\{ \pi\;|\;\trace(\pi)\models\varphi\}$.

\begin{definition}[Realizability]
  Given a temporal formula $\varphi$ over propositions $\MX$ and $\MY$,
  the formula is called realizable if Sally wins the game
  $(\A(\MX,\MY),W_\varphi)$.
\end{definition}

We now create a safety game defining a set of bad positions $B$ of the
arena, and show in Lemma~\ref{lem-safetygames} that the resulting
safety game solves the realizability problem.
\begin{align*}
B = \{ (\overline{x},\overline{y}) \;|\; 
& 
\text{there exists } v\in\Val(\MX) \text{ such that for all } v'\in\Val(\MY):\\
& 
\overline{x}\cdot v+\overline{y}\cdot v' \not\wdash\alpha\And\Always\psi\}.
\end{align*}
Even though we do not show it here, it is easy to see that in order to
determine whether a state $(\overline{x},\overline{y})$ of $P_E$  is bad, it is
sufficient to remember (1) whether some prefix was bad (for which only
one bit is necessary) and (2) the last $m$ elements of $(\overline{x},\overline{y})$ where
$m$ is the maximum number of nested $\Next$ operators in the
specification.
Hence, the arena of a safety specification can be turned into a finite
arena.
Therefore the realizability of safety specifications is decidable (by
reduction into whether Sally has a winning strategy in a finite state
safety game).
Given a specification $\varphi$ we call $\G(\varphi)$ for the safety
specification game for $\varphi$.

   \begin{lemma}\label{lem-safetygames}
     Given a safety specification $\varphi$, $\varphi$ is realizable
     if and only if $\G(\varphi)$ is winning for Sally.
   \end{lemma}

   \begin{proof}
     Assume that $\varphi$ is realizable, and let $\rho_S$ be a
     strategy for Sally in the specification game.
     Therefore any play $\pi$ played according to $\rho_S$ is winning
     for Sally, so $\pi\in W_\varphi$ i.e. $\trace(\pi)\models\varphi$.
     Hence, by Corollary~\ref{cor:witness} there is no finite witness
     of the violation of $\trace(\pi)$, consequently $\rho_S$ is a
     winning strategy for Sally in $\G(\varphi)$.
     Similarly, assume that Sally wins $\G(\varphi)$ and let $\pi$ be
     an arbitrary play played according to a winning strategy $\rho_S$.
     The play $\pi$ in the specification game is winning for Sally as
     well, due to Corollary~\ref{cor:witness}.
     \qed
 \end{proof}

%

%% file: tnf.tex
\section{Terse Normal Form}
\label{sec:tnf}
Our tableaux for a safety specification $\varphi = \alpha\And\Always\psi$ analyze its realizability on the basis of a play where the environment chooses a move on its variables $\MX$ and, then, the system chooses its move on its variables $\MY$ according to the safety specification. In order to the branches of the tableau to represent the real play, any formula in a node should determine the true strict-future possibilities of the game.
For example, considered that following formula in {\em disjunctive
  normal form} (\DNF) is the representation of all possible moves at
some state of the game:
\begin{equation}
\label{eq-DNF}
(\Next\Not s)\Or(p_e\And s\And\Next\Next s)
\end{equation}
This means that, after any choice of the environment and the system, satisfying $(\Next\Not s)$ would fulfill  the specification.
Moreover, after the environment moves $p_e$ followed by the system
moving $s$, not only $\Next\Next s$ would be coherent with the
specification, but also $\Next\Not s$.
However, a classical tableau-style analysis would split (\ref{eq-DNF})
into two branches such that the one containing $p_e\And s$ requires
$\Next\Next s$ to satisfy the specification, precluding the
possibility of $\Next\Not s$.
Note also that the formula
\begin{equation}
\label{eq-TNF0}
(p_e\And s\And(\Next\Not s\Or\Next\Next s))\Or (\Not p_e\And \Next\Not s)\Or (\Not s\And \Next\Not s)
\end{equation}
is logically equivalent to (\ref{eq-DNF}), but suitable for a tableau-style analysis of realizability. In this section, we define the {\it Terse Normal Form} ($\TNF$) for safety formulas that allows us to associate to any move the formula that any trace must satisfy in the (strict) future to be coherent with the current safety specification. The formula (\ref{eq-TNF0}) is in $\TNF$.

We consider a partition of all the {\it basic} (sub)formulas that
occurs in a safety formula---that is all the formulas of the forms
$\ell$, $\Next^n\eta$, $\Eventually_I\eta$ or $\Always_I\eta$---into
two classes of formulas:
\begin{compactitem}
\item the class {\it from-now}: $\ell, \Eventually_{[0,m]}\eta, \Always_{[0,m]}\eta$
\item the class {\it from-next}: $\Next\eta$,
  $\Eventually_{[n,m]}\eta$ and $\Always_{[n,m]}\eta$ (for any  $n\geq 1$).
\end{compactitem}

\begin{definition}[Strict-future and separated]
\label{def-strict}
A safety formula is a {\it strict-future formula} if and only if it is
a \DNF combination of from-next formulas.
A safety formula is a {\it separated formula} if and only if it is the
(possibly empty) conjunction of a set of Boolean literals, denoted as $\Lit(\pi)$, and (at most) a strict-future formula, denoted as $\Fut(\pi)$.
\end{definition}
%
%
\begin{definition}[TNF]
\label{def-TNF}
A safety formula $\eta$ is in {\it Terse Normal Form} $(\TNF)$ if and
only if it is a disjunction $\bigvee_{i=1}^n\pi_i$ such that each
$\pi_i$ is a separated formula, and for all $1\leq i\neq j\leq n$
there is at least one literal $\ell$ such that $\ell\in\Lit(\pi_i)$ and
$\Not \ell\in\Lit(\pi_j)$.
\end{definition}

\newcounter{prop-TNF-equiv}
\setcounter{prop-TNF-equiv}{\value{proposition}}

\begin{proposition}
\label{prop-TNF-equiv}
For any safety formula $\eta$ there is a logically equivalent safety formula, called $\TNF(\eta)$, that is in \TNF.
\end{proposition}
\begin{proof} 
First, any safety formula $\eta$ (for simplicity, suppose that $\eta$ is in \NNF) can be converted into a disjunctive normal form-like formula, that we call $\DNF(\eta)$, whose ``literals'' are either classical literals ($p$ or $\Not p$), or from-next formulas. 
For constructing $\DNF(\eta)$, it suffices to use, besides
classical logical equivalences on boolean connectives, the following equivalences on temporal connectives:
$$
\begin{array}{ll}
\Eventually_{[n,n]}\beta\equiv\Next^n\beta &  \hspace{1cm}
\Always_{[n,n]}\beta\equiv\Next^n\beta\\
\Eventually_{[n,m]}\beta\equiv\Next^n\beta\Or\Next\Eventually_{[n..m-1]}\beta & \hspace{1cm}
\Always_{[n,m]}\beta\equiv\Next^n\beta\And\Next\Always_{[n..m-1]}\beta\\
\end{array}
$$
Then, we transform each pair of disjuncts in $\DNF(\eta)$ with indexes  $1\leq i\neq j\leq n$ such that 
for all literal $\ell\in\Lit(\pi_i)$ it holds that $\Not \ell \not\in\Lit(\pi_j)$ as follows. Let $\delta = \Lit(\pi_i)\cap\Lit(\pi_j)$, $\delta_1 = \Lit(\pi_i)\setminus\delta$ and 
$\delta_2 = \Lit(\pi_j)\setminus\delta$.
Then, we apply
\begin{eqnarray}
\nonumber
(\delta\And\delta_1\And\eta_1) \Or  (\delta\And\delta_2\And\eta_2) & \equiv &
(\delta\And\delta_1\And\delta_2\And(\eta_1\Or \eta_2))  \\
\label{TNF-gen-equiv}
& &
\Or\,\DNF(\delta\And\delta_1\And\Not\delta_2\And \eta_1)\\
\nonumber
& & 
\Or\DNF(\delta\And\Not\delta_1\And\delta_2\And\eta_2)
\end{eqnarray}
where $\Fut(\pi_1) = \eta_1$ and $\Fut(\pi_j) = \eta_2$.
Note that, in particular cases, $\delta,\delta_1$ or $\delta_2$ could be empty (equivalently $\cT$), hence their negation are equivalent to $\cF$. In this cases, the equivalence (\ref{TNF-gen-equiv}) could be simplified. For example, if $\Lit(\pi_i)=\Lit(\pi_j)$, then $\delta= \Lit(\pi_i)\cap\Lit(\pi_j) = \Lit(\pi_i) = \Lit(\pi_j)$. Consequently, (\ref{TNF-gen-equiv}) is simplified to
$(\delta\And\eta_1) \Or  (\delta\And\eta_2) \equiv \delta\And(\eta_1\Or \eta_2)$.

This transformation is repeatedly applied until every pair $(\pi_i,\pi_j)$ satisfies the required condition. 
It is easy to see that the above process produces a formula in \TNF.
Moreover, since we only apply logical equivalences to subformulas, by
substitutivity, the resulting formula $\TNF(\eta)$ is logically
equivalent to $\eta$.\qed
\end{proof}
\begin{example}
  \label{ex-TNF1}
  The \TNF for $p_e \Iff \Next s$ and $\Next p_e \Iff \Next s$ from
  Section \ref{sec:intro} are:
  \vspace{-2mm}
  \begin{align*}
    \TNF(p_e \Iff \Next s) &= (p_e\And \Next s) \Or (\Not p_e\And \Next\Not s)\\
    \TNF(\Next p_e \Iff \Next s) & = (\Next p_e\And \Next s) \Or (\Next\Not p_e\And \Next\Not s)\\[-8mm]
  \end{align*}
Note that $\TNF(p_e \Iff \Next s)$ is composed by two disjuncts,
each having a literal and a strict-future formula, but
$\TNF(\Next p_e \Iff \Next s)$ has only one move (the empty set of
literals) with one future-strict formula (which is a disjunction).
Finally, for
$ \eta = c \And (\Not p_e \Impl \Always_{[0,9]} \Not c) \And
(\Always_{[0,9]} c \Or \Eventually_{[0,2]}\Not c) $:
\vspace{-2mm}
\begin{eqnarray*}
\TNF(\eta) & = & 
(p_e\And c \And (\Next\Eventually_{[0,1]} \Not c \Or\Next\Always_{[0,8]} c)) \Or (\Not p_e \And c \And\Next\Always_{[0,8]} c).
\end{eqnarray*}
\end{example}
For any formula $\bigvee_{i=1}^n\pi_i$ in $\TNF$ and any
$1\leq i\leq n$, we denote by $\Lit(\pi_i)$ the conjunction (or set)
of literals in $\pi_i$ and by $\Fut(\pi_i)$ the unique strict-future
formula in $\pi_i$.  Abusing language, we say that each $\pi_i$ is a
{\it move}, though it could really represent a collection of pairs of
moves (one of the environment followed by one of the system) because
the variables that do not appear in $\Lit(\pi_i)$ could be freely
chosen.
Note that $\Val_{\pi_i}=\Val_{\Lit(\pi_i)}$ for any move
$\pi_i$ of any formula in $\TNF$.

\begin{proposition}
\label{prop-TNF}
Let $\eta$ be a safety formula and $\TNF(\eta) = \bigvee_{i=1}^n\pi_i$
\begin{enumerate}[\textup{(}a\textup{)}]
\item For any trace $\sigma$, $\sigma\models\eta$ iff
  $\sigma\models\pi_i$ for exactly one $1\leq i\leq n$.
\item For any finite trace $\lambda$, $\lambda\wdash\eta$ iff
  $\lambda\wdash\pi_i$ for exactly one $1\leq i\leq n$.
\item Let $\sigma$ be such that $\sigma\models\eta$ and let $1\leq i\leq n$. Then, $\sigma\models\Lit(\pi_i)\Impl\Fut(\pi_i)$.
\end{enumerate}
\end{proposition}
\begin{proof}
Items ($a$) and ($b$) are easy consequences of Definition \ref{def-TNF} and the fact that $\eta\equiv\TNF(\eta)$ (Proposition \ref{prop-TNF-equiv}). 
For item ($c$), consider any trace $\sigma$ and any $1\leq i\leq n$ such that $\sigma\models\eta$ and $\sigma\models\Lit(\pi_i)$. Then, $\sigma_0\in\Val_{\Lit(\pi_i)}(\MX\cup\MY)$. Let us suppose that $\sigma\not\models\Fut(\pi_i)$ then, by ($a$), it should exists $1\leq j\leq n$ such that $i\neq j$, $\sigma\models\pi_j$. Then, $\sigma\models\Lit(\pi_j)$ and $\sigma_0\in\Val_{\Lit(\pi_j)}(\MX\cup\MY)$. The latter contradicts $\sigma_0\in\Val_{\Lit(\pi_i)}(\MX\cup\MY)$, because there is a literal $\ell$ such that $\ell\in\Lit(\pi_i)$ and $\Not \ell\in\Lit(\pi_j)$.\qed
\end{proof}

Many moves in formulas in \TNF could be discharged for realizability test purposes. In a practical implementation of our method, we will discharge them. However, for simplicity of the presentation and the proofs of soundness and completeness, we consider all the moves in the \TNF of a formula.

From now on we fix a safety specification on variables $\MX\cup\MY$ whose safety formula is $\psi$.
Consider any formula $\bigvee_{i=1}^n\pi_i$ in \TNF. By Proposition \ref{prop-TNF}, the collection $\{\Val_{\pi_i}(\MX\cup\MY)\mid 1\leq i\leq n\}$ is pairwise disjoint. Next, we define properties of the collection $\{\Val_{\pi_i}(\MX)\mid 1\leq i\leq n\}$.

\begin{definition}
\label{def-covering}
A formula $\bigvee_{i=1}^n\pi_i$ in \TNF is an {\it $\MX$-covering} if and only if
$$\bigcup_{i=1}^n\Val_{\pi_i}(\MX) = \Val(\MX).$$
A formula $\bigvee_{i=1}^n\pi_i$ in \TNF is a {\it minimal $\MX$-covering} if it is a $\MX$-covering and $\bigvee_{i=1,i\neq j}^n\pi_i$ is not an $\MX$-covering for every $1\leq j\leq n$.
\end{definition}

Intuitively, a minimal $\MX$-covering represents a system strategy
from the current position. Therefore, the collection of all minimal
coverings represents all possible strategies. Moreover, each move in
an strategy contains all the strict-future possibilities for this
move.

\begin{example}
\label{ex-TNF4}
Let
$\TNF(\eta) = (p_e\And c \And \eta_1) \Or (\Not p_e\And c
\And\eta_2)\Or (\Not c\And\eta_3)$ where $\eta_1,\eta_2,\eta_3$ are
strict-future formulas and $\MX=\{p_e\}$. It is a non-minimal
$\MX$-covering, but the third move $\Not c\And\eta_3$ is a minimal
$\MX$-covering, The two first moves together also provide a minimal
$\MX$-covering.
\end{example}

\begin{example}
\label{ex-TNF5}
Let $\TNF(\eta) =
(p_e\And c \And \eta_1) \Or (\Not p_e\And q_e\And c \And\eta_2)\Or (\Not c\And\eta_3)$ where $\eta_1,\eta_2,\eta_3$ are strict-future formulas and $\MX=\{p_e,q_e\}$. $\TNF(\eta)$ is a non-minimal $\MX$-covering, but the third move $\Not c\And\eta_3$ is a minimal $\MX$-covering, However, the disjunction of the two first moves is not an $\MX$-covering, because the valuation that makes both environment variables to be false is not included there.
\end{example}

\begin{example}
\label{ex-TNF6}
The following two \TNF formulas (where $\eta_1,\eta_2,\eta_3,\eta_4$ are strict-future formulas and $\MX=\{p_e\}$) are each composed by four different minimal $\MX$-coverings:
$$
(p_e\And c \And \eta_1) \Or (\Not p_e\And c \And\eta_2)\Or (p_e\And \Not c \And \eta_3) \Or (\Not p_e\And \Not c \And\eta_4)
$$
$$
(p_e\And c \And \eta_1) \Or (\Not p_e\And \Not c \And\eta_2)\Or (p_e\And \Not c \And d \And \eta_3) \Or (\Not p_e\And c\And \Not d \And\eta_4).
$$
\end{example}

Abusing language, we say that a set of indices $I$ is a (minimal) $\MX$-covering when really the  formula $\bigvee_{i\in I}\pi_i$ is a (minimal) $\MX$-covering.

\begin{proposition}
\label{prop-covering}
Let $\Phi$ any set of safety formulas and $\TNF(\Phi\And\psi) = \bigvee_{i\in I}\pi_i$.
\begin{enumerate}[\textup{(}a\textup{)}]
\item 
\label{i}
If $I$ is not an $\MX$-covering, then there exists some $v\in\Val(\MX)$ such that $v\not\wdash\Phi\And\psi$.
\item 
\label{ii}
If $I$ is a minimal $\MX$-covering, then for all $i\in I$ and all  $v\in\Val_{\pi_i}(\MX)$, there exists some $v'\in\Val_{\pi_i}(\MY)$ such that $v+v'\in\Val_{\pi_i}(\MX\cup\MY)$.
\item 
\label{iii}
If for each $v\in\Val(\MX)$ there exists $v'\in\Val(\MY)$ such that $v + v'\wdash\Phi\And\psi$, then there exists some minimal $\MX$-covering $J\subseteq I$.
\end{enumerate}
\end{proposition}
\begin{proof}
For item ($\ref{i}$), there exists $v\in\Val(\MX)$ such that $v\not\wdash\pi_i$ for all $1\leq i\leq n$.  By Proposition \ref{prop-TNF}, $v\not\wdash\Phi\And\psi$. Items ($\ref{ii}$) and ($\ref{iii}$), are easy consequences of Definition \ref{def-covering}.\qed
\end{proof}

To handle strict-future formulas $\Fut(\pi)$ in the tableau rules we
introduce the symbol $\Ddot{\Or}$ which is semantically equivalent to
$\Or$, but our tableau rules deal differently with both disjunctive operators.
More precisely, strict-future subformulas $\Fut(\pi)$ (inside moves of \TNF formulas) will be written as $\Ddot{\bigvee}_{i=1}^m\delta_i$.


%% file: tableaux.tex
\section{Realizability Tableaux}
\label{sec:tableaux}

Our tableaux are AND-OR trees of nodes, each labelled by a set of formulas. Hence, each node has 0, 1 or more AND-successors or OR-successors.
A node is said to be the parent of its successors nodes.
In the examples below we mark AND-nodes with a semicircle embracing
all the edges to the AND-successors of a node. A tableau is
constructed for an input safety specification that is the root of the
tree. The set of tableau rules (see Figures \ref{fig-always-rules},
\ref{fig-saturation-rules} and \ref{fig-next-state-rule}) determine
how the tableau can be constructed. Each rule determines the labels on
the children of a node and the kind (AND or OR) of its successors.  A
tableau is completed (or finished) when no further rule can be applied. Rules
can be applied only to nodes in branches that are neither failed nor
successful. A node is called a leaf when no rule can be applied to
it. There are two kinds of leaves. Failure leaves are labelled by
(syntactically) inconsistent sets of formulas, which indicates that the
branch from the root to the leaf is failed. Successful leaves are
labelled by sets of formulas that are subsumed (in the sense we will
precise below) by some previous node in the branch from the root to
the leaf.

The following definition formalizes our notion of tableau in terms of
many concepts that will be precised in the next subsections.

\begin{definition}
\label{def-tableau}
A {\it tableau} for a safety specification $\varphi = \alpha\And\Always\psi$ is a labelled tree $\Tableau(\varphi) = (N, \tau, R)$, where $N$ is a set of nodes, $\tau$ is a mapping of the nodes of $T$ to formulas in $\Clo(\varphi)$ and $R\subseteq N\times N$, such that the following conditions hold:
\begin{itemize}
\item 
The root is labelled by the set $\{\alpha,\Always\psi\}$.
\item 
For any pair of nodes $(n,n')\in R$, $\tau(n')$ is the set of formulas obtained as the result of the application of one of the tableau rules (in Figures \ref{fig-always-rules}, \ref{fig-saturation-rules} and \ref{fig-next-state-rule}) to $\tau(n)$.  Given the applied rule is $\rho$, we term $n'$ a $\rho$-successor of $n$
\item For every success or failure leaf $n$ there is no $n'\in N$ such that
  $(n,n')\in R$ where:
\begin{itemize}
    \item 
A failure leaf is a node such that $n\in N$ such that $\Incons(\tau(n))$ (see Definition \ref{def-inconsistent}).
\item
A success leaf is a node $n\in N$ such that $\Always\psi\in\tau(n)$ and there exists $k\geq 0$, $n_0,\dots,n_k\in N$
such that $(n_i,n_{i+1})\in R$ for all $0\leq i <k$, $(n_k,n)\in R$
and $\tau(n_0)\lessdot\tau(n)$ (see Definition \ref{def-lessdot}).
\end{itemize}
\end{itemize}
\end{definition}
For any tableau rule $\rho$, the set of $\rho$-successors of the node to which $\rho$ is applied is the set of children generated by $\rho$. Note that the $\rho$-successors of a node are OR-siblings for any rule $\rho$ except for the rule $(\Always\&)$ which are AND-siblings. In examples, we use an arc for linking the edges of the nodes that are AND-siblings, and no mark in the edges of OR-siblings.

\subsection{Subsumption and Syntactical Inconsistency}
The set of formulas used to label our tableaux nodes are subsumption-free with respect to the collection of subsumption rules between temporal formulas given in Definition \ref{def-subsumption} and classical subsumption on boolean formulas. 
Let $\beta\sqsubseteq\gamma$ denote the fact that $\beta$ {\it subsumes} $\gamma$ or that $\gamma$ {\it is subsumed by}~$\beta$. Subsumption is related to logical implication or logical consequence in the sense that,
if $\beta \sqsubseteq \gamma$, then 
$\models\beta \rightarrow \gamma$ or equivalently $\beta\models\gamma$. Note that the converse is not necessarily true. Classical subsumption rules includes $\beta\sqsubseteq\beta$, $\beta\And\gamma\sqsubseteq\beta$, and $\beta\sqsubseteq\beta\Or\gamma$.
\begin{definition} 
\label{def-subsumption}
The subsumption rules for temporal formulas that apply in our tableau method are:
\begin{itemize}
\item For all $n\leq n'$ and $m'\leq m$\footnote{By construction of our tableaux, we always generate from formulas with starting interval at $n$ new formulas were $n' = n$. Note also that $\Next^n\beta = \Eventually_{[n,n]} \beta = \Always_{[n,n]} \beta$.},\\ $\Eventually_{[n',m']}\beta \sqsubseteq \Eventually_{[n,m]}\beta$, $\Always_{[n,m]}\beta \sqsubseteq \Always_{[n',m']}\beta$,
and $\Always_{[n',m']}\beta \sqsubseteq\Eventually_{[n,m]}\beta$.
\item For all $n\leq k\leq m$: $\Next^{k}\beta \sqsubseteq\Eventually_{[n,m]}\beta$ and $\Always_{[n,m]}\beta\sqsubseteq \Next^{k}\beta$.
\end{itemize}
\end{definition}
The following result easily follows from Def. \ref{def-subsumption} and semantics.
\begin{proposition}
\label{prop-subsumption}
  Let $\beta\sqsubseteq\gamma$ be a pair of formulas.
  For any trace $\sigma$, if $\sigma\models\beta$ then
  $\sigma\models\gamma$.
  For any finite trace $\lambda$, if $\lambda\wdash\beta$ then
  $\lambda\wdash\gamma$.
  Consequently, $\sigma\not\models\beta\And\widetilde{\gamma}$
  and $\lambda\not\wdash\beta\And\widetilde{\gamma}$ for any 
  $\sigma$ and $\lambda$.
\end{proposition}

According to Proposition \ref{prop-subsumption},
subsumption is used to detect (syntactical) inconsistencies (see Definition \ref{def-inconsistent}(\ref{b1})). Inconsistencies are used to close tableau branches by a failure leaf. No rule is applied to a node labelled by an inconsistent set.

\begin{definition} 
\label{def-inconsistent}
A set of formulas $\Phi$ is {\em (syntactically) inconsistent} (denoted $\Incons(\Phi)$) 
if and only if one of the following four conditions hold:
\begin{enumerate}[(a)]
    \item $\cF\in\Phi$
    \item
    \label{b1}
    $\{\beta, \widetilde{\gamma}\}\subseteq\Phi$ for some $\beta,\gamma$ such that $\beta\sqsubseteq\gamma$ 
    \item $\{x = c_1, x = c_2\}\subseteq\Phi$ for some $c_1\neq c_2$
    \item $\{ \Not(x = c) \mid c\in T\}\subseteq\Phi$ for some enumerated type $T$.
\end{enumerate}
Otherwise,  $\Phi$ is {\em (syntactically) consistent}, denoted $\Cons(\Phi)$.
\end{definition}
\begin{example}
\label{ex-inconsistent}
The set $\{\Always_{[2, 8]} c , \Next^5 \Not c \}$ is inconsistent because
$\Always_{[2, 8]} c \sqsubseteq \Next^5  c$ and $\Next^5 \Not c = \widetilde{\Next^5  c} $.
\end{example}

Inconsistencies are used to close tableau branches. No rule is applied to a node labelled by an inconsistent set, and this node is called a failure leaf.

%
We define the following subsumption-based order relation between sets of formulas for detecting successful leaves.
\begin{definition}
\label{def-lessdot}
For two given set of formulas $\Phi$ and $\Phi'$, we say that $\Phi \lessdot \Phi'$ if and only if for every formula $\beta\in\Phi$ there exists some $\beta'\in\Phi'$ such that  $\beta\sqsubseteq\beta'$. For two given strict-future formulas,
$\Ddot{\bigvee}_{i=1}^n\Delta_i\sqsubseteq\Ddot{\bigvee}_{j=1}^m\Gamma_j$ if and only if for all $1\leq i\leq n$ there exists $1\leq j\leq m$ such that $\Delta_i\lessdot\Gamma_j$.
\end{definition}

\begin{proposition}
\label{prop-lessdot}
For any finite trace $\lambda$ and any pair of set of formulas $\Phi$ and $\Phi'$ such that $\Phi\lessdot \Phi'$, if $\lambda\wdash\Phi$ then $\lambda\wdash\Phi'$.
\end{proposition}
\begin{proof}
By Definition \ref{def-lessdot} and Proposition \ref{prop-subsumption}.\qed
\end{proof}

No rule is applied to a node that is labelled by a set $\Phi'$, such that for some previous label $\Phi$ in the same branch (some path from the root) it holds that 
$\Phi\lessdot\Phi'$ (see Definition \ref{def-tableau}).

\subsection{Tableau Rules}
In this section, we introduce the tableau rules, along with the concepts and properties related with the sets of formulas generated from a safety specification according to this set of  rules. 

\begin{figure}[t!]
\centering
\[
\begin{array}{lll}
(\Always\cF) &  \displaystyle{\frac{ \Phi,\Always \psi}
{\cF, \Always\psi}} &\mbox{ if } \tau \mbox{ is not an $\MX$-covering} \\[4mm]
(\Always\|) & \displaystyle{\frac{\Phi,\Always \psi}
{\bigvee_{i\in J_1}\pi_i,\Always\psi \;\mid\dots\mid\; \bigvee_{i\in J_m}\pi_i,\Always\psi}} & \mbox{ if } J_1,\dots,J_m \mbox{ is the collection of } \\
 & & \mbox{ all minimal $\MX$-covering of } \tau\\[2mm]
(\Always\&) & \displaystyle{\frac{ \bigvee_{i\in I}\pi_i,\Always \psi}{ \pi_{1}, \Next\Always\psi \;\;\&\; \ldots \;\;
\&\; \pi_{n}, \Next\Always\psi}}  & \mbox{ if } I \mbox{ is a minimal  $\MX$-covering}\\
\end{array}
\]
\vspace{-5mm}
\caption{Always Rules (where $\tau$ denotes $\TNF(\Phi\And\psi)$)}
\label{fig-always-rules}
\end{figure}

First, the {\it Always Rules} in Figure \ref{fig-always-rules}
provides a non-deterministic procedure of analyzing the minimal
$\MX$-coverings in the $\TNF(\Phi\And\psi)$ (see Definition
\ref{def-covering} and Proposition \ref{prop-covering}). The rule
$(\Always\&)$ is the only rule in our system that produces
AND-successors for splitting the cases of each minimal $\MX$-covering.

\begin{figure}[t!]
\centering
$$
\begin{array}{ll}
(\Or)\; \displaystyle{\frac{ \Phi, \beta\Or\gamma}{ \Phi, \beta \;\mid\; \Phi,\gamma}}
 & \hspace{1.5cm}
(\Ddot{\Or}\Or)\;\displaystyle{\frac{ \Phi, (\eta\wedge(\beta\Or\gamma))\Ddot{\Or}\delta }{ \Phi, (\eta\And\beta) \Ddot{\Or} (\eta\And \gamma)\Ddot{\Or}\delta}}
\\[5mm]
(\And) \;\displaystyle{\frac{ \Phi, \beta\And\gamma}{ \Phi, \beta, \gamma}} 
& 
\\[4mm]
\multicolumn{2}{l}{
(\Eventually\!<) \; \displaystyle{\frac{ \Phi,\, \, \Eventually_{[n,m]}\beta}{ \Phi,\,   \;  \Next^n\beta\;\;\mid \;\;\Phi, \, \Next\Eventually_{[n,m-1]}\beta}}\;\;\mbox {if } n < m} 
\\[4mm]
 \multicolumn{2}{l}{(\Ddot{\Or}\Eventually\!<) \; \displaystyle{\frac{ 
 \Phi,\, (\eta\And\Eventually_{[n,m]}\beta)\Ddot{\Or}\delta}
{ \Phi,\,   
(\eta\And\Next^n\beta)\Ddot{\Or}(\eta\And\Next\Eventually_{[n,m-1]}\beta\}\Ddot{\Or}\delta}} \;\;\mbox {if } n < m}
\\[4mm]
(\Eventually\!=) \;\displaystyle{\frac{\Phi, \, \Eventually_{[n,n]}\beta}{ \Phi, \,\Next^n\beta}} 
& \hspace{1.5cm}
(\Ddot{\Or}\Eventually\!=) \;\displaystyle{\frac{\Phi,\, (\eta\And\Eventually_{[n,n]}\beta)\Ddot{\Or}\delta}{ \Phi,\, (\eta\And\Next^n\beta)\Ddot{\Or}\delta}}
\\[4mm]
(\Always\!<) \; \displaystyle{\frac{ \Phi,\,\Always_{[n,m]}\beta}{ \Phi,\, \Next^n\beta,\,  \Next\Always_{[n,m-1]}\beta}} \;\;\mbox {if } n < m  
& \hspace{1.5cm}
(\Always\!=) \; \displaystyle{\frac{\Phi,\, \Always_{[n,n]}\beta}{ \Phi,\,\Next^n\beta}}
\\[5mm]
\multicolumn{2}{l}{(\Ddot{\Or}\Always\!<) \; \displaystyle{\frac{ \Phi,\,(\eta\And\Always_{[n,m]}\beta)\Ddot{\Or}\delta}
{\Phi,\, (\eta\And\Next^n\beta\And \Next\Always_{[n,m-1]}\beta) \Ddot{\Or}\delta }} \;\;\mbox {if } n < m}  
\\[4mm]
\multicolumn{2}{l}{(\Ddot{\Or}\Always\!=) \; \displaystyle{\frac{\Phi,\,(\eta\And\Always_{[n,n]}\beta)\Ddot{\Or}\delta}
{ \Phi,\,(\eta\And\Next^n\beta)\Ddot{\Or}\delta}}
}
\end{array}
$$
\vspace{-5mm}
\caption{Saturation Rules}
\label{fig-saturation-rules}
\end{figure}
Next, we introduce the set of rules that are used to decompose
formulas into its constituents, in the usual way that tableau methods
perform the so-called {\it saturation}.
In our method, decomposition of formulas inside the conjunction (or sets) connected by the operator $\Ddot{\Or}$ just performs an unfolding in the formula. 
The {\it Saturation Rules} in Figure~\ref{fig-saturation-rules} are used to saturate with respect to classical connectives $\And$ and $\Or$ (including $\Ddot{\Or}$) and temporal operators $\Eventually_I$ and $\Always_I$.

\begin{proposition}
\label{prop-correct-rules}
\leavevmode
For any saturation rule {\Large $\frac{\Phi}{\Phi_1\mid\cdots\mid\Phi_k}$}, it holds that
$\sigma\models\Phi$ if and only if $\sigma\models\Phi_i$ for some $1\leq i\leq k$.
\end{proposition}
\begin{proof}
By routinely applying semantics.\qed
\end{proof}

\begin{definition}
\label{def-elementary-strict-future}
A next-formula is a formula whose first symbol is $\Next$.
A strict-future formula $\Ddot{\bigvee}_{i=1}^n\Delta_i$ is {\it elementary} if every formula in the set $\bigcup_{i=1}^n\Delta_i$ is a next-formula.
\end{definition}

\begin{proposition}
\label{prop-elementary-strict-future}
For any given strict-future formula $\delta$, there exists an elementary formula that we call $\delta^E$ such that $\delta\equiv\delta^E$ and $\delta^E$ is in \DNF.
\end{proposition}
\begin{proof}
Repeatedly apply to $\delta$ the rules $(\Ddot{\Or}\And)$, $(\Ddot{\Or}\Eventually <)$, $(\Ddot{\Or}\Eventually =)$, $(\Ddot{\Or}\Always <)$ and $(\Ddot{\Or}\Always =)$, until no one can be applied. \qed
\end{proof}

\begin{definition}
A set of formulas $\Delta$ is {\em saturated} if and only if for all $\delta\in\Delta$ the following conditions holds:
\begin{compactitem}
    \item If $\delta = \beta \And \gamma$, then $\{\beta,\gamma\}\in \Delta$.
    \item If $\delta = \beta \vee \gamma$, then $\beta\in \Delta$ or $\gamma\in \Delta$.
    \item If $\delta = \Always_{[n,m]}\beta$ and $n < m$, then $\{\Next^n\beta, \Next\Always_{[n,m-1]}\beta\}\subseteq\Delta$.
     \item If $\delta = \Eventually_{[n,m]}\beta$ and $n < m$, then either $\Next^n\beta\in\Delta$.
    or $\Next\Eventually_{[n,m-1]}\beta\in\Delta$    
    \item If $\delta = \Always_{[n,n]}\beta$ or $\gamma = \Eventually_{[n,n]}\beta$, then $\Next^n\beta\in\Delta$.
    \item If $\delta$ is a strict-future formula, then $\delta^E\in\Delta$
    \end{compactitem}
We use $\Stt(\Delta)$ to denote the set of all (minimal) saturated sets that contains $\Delta$.
\end{definition}

\begin{proposition}
\label{prop-Stt}
Let $\Delta$ be a set of formulas,  $\sigma$ a trace and $\lambda$ a finite trace.
\begin{itemize}
\item $\sigma\models\Delta$ if and only if $\sigma\models\Phi$ for
  some $\Phi\in\Stt(\Delta)$.
\item $\lambda\wdash\Delta$ if and only if $\lambda\wdash\Phi$ for
  some $\Phi\in\Stt(\Delta)$.
\end{itemize}
\end{proposition}
\begin{proof}
By induction on the construction of $\Phi\in\Stt(\Delta)$.\qed
\end{proof}

\begin{proposition}
\label{prop-always-rule}
Let  $\Phi$ be a set of safety formulas and let $J_1,\dots, J_m$ the collection of all minimal $\MX$-coverings in $\TNF(\Phi\And\psi) = \bigvee_{i\in I}\pi_i$. Then
\begin{enumerate}[(a)]
\item \label{a}
For any trace $\sigma$,
  $\sigma\models\Phi,\Always\psi$ iff
  $\sigma\models \pi_{i}, \Next\Always\psi$ holds for some $i\in J_k$
  for each $1\leq k\leq m$.
  Let $\lambda$ be finite trace, $\lambda\wdash\Phi\And\psi$ iff
  $\lambda\wdash\pi_{i}$ for some $i\in J_k$ for each $1\leq k\leq m$.
\item \label{b}
For any $1\leq k\leq m$ and any $i\in J_k$ the following two facts hold:
\begin{enumerate}[\textup{(}i\textup{)}]
\item If $\Incons(\Delta)$ for all $\Delta\in\Stt(\Phi\cup\{\pi_{i}\})$, then every $\lambda_0\in\Val_{\Delta}(\MX\cup\MY)$ is a witness of the violation of the safety specification $\Phi\And\Always\psi$. 
\item For any $\Delta\in\Stt(\Phi\cup\{\pi_{i}\})$ such that $\Cons(\Delta)$, it holds that $\lambda_0\wdash\Phi\And\psi$ for every $\lambda_0\in\Val_{\Delta}(\MX\cup\MY)$. 
\end{enumerate}
\end{enumerate}
\end{proposition}
\begin{proof}
Item (\ref{a}) follows from Proposition \ref{prop-TNF}. Item (\ref{b}) follows from item (\ref{a}) and Proposition \ref{prop-Stt}.\qed
\end{proof}

In the following two results, under certain conditions on $\TNF(\Phi\And\psi)$, we get some valuation $v\in\Val(\MX)$ such that $v + v'\not\wdash\Phi\And\psi$ holds for all $v'\in\Val(\MY)$.

\begin{proposition}
\label{prop-non-x-covering}
Let $\Phi$ be any set of safety formulas. If $\TNF(\Phi\And\psi)$ is not an $\MX$-covering, then there exists some $v\in\Val(\MX)$ such that for all $v'\in\Val(\MY)$,
$v + v'\not\wdash\Phi\And\psi$
\end{proposition}
\begin{proof}
By Definition \ref{def-covering}, there is some $v\in\Val(\MX)\setminus\Val_{\Phi\wedge\psi}(\MX)$. \qed
\end{proof}

\begin{proposition}
\label{prop-all-minimal-x-covering}
Let $\Phi$ be any set of safety formulas and $\TNF(\Phi\And\psi)= \bigvee_{i\in I}\pi_i$ be an $\MX$-covering and let $J_1,\dots, J_m$ ($m\geq 1$) the collection of all minimal $\MX$-coverings in $I$.
If for every $1\leq k\leq m$ there exists some $i\in J_k$ such that $\Incons(\Delta)$ for all $\Delta\in\Stt(\pi_i)$, then there exists some $v\in\Val(\MX)$ such that for all $v'\in\Val(\MY)$,
$v + v'\not\wdash\Phi\wedge\psi$.
\end{proposition}
\begin{proof}
Suppose that for all $v\in\Val(\MX)$  there is some $v'\in\Val(\MY)$ such that  $v + v'\wdash\Phi\And\psi$. Then, by Proposition \ref{prop-covering}(\ref{iii}), there exists some $1\leq k\leq m$ such that the minimal $\MX$-covering $J_k\subseteq I$. Therefore, by Proposition \ref{prop-covering}(\ref{ii}), for all $v\in\Val(\MX)$  there is some $j\in J_k$ and some $v'\in\Val(\MY)$ such that  $v + v'\wdash\pi_j$. According to Proposition \ref{prop-Stt}, for each $j\in J_k$ there exists some $\Delta\in\Stt(\pi_j)$ such that $\Cons(\Delta)$, which contradicts the hypothesis.\qed
\end{proof}

Next, we introduce the {\it next-state rule} that (roughly speaking) allows us to jump from one state to the next one. For that, we introduce some preliminary concepts and notation.

\begin{definition}
\label{def-elementary-set}
A set of formulas $\Phi$ is {\it elementary} if it consists of a set of literals and one elementary strict-future formula.
\end{definition}
\begin{definition}
\label{def-downarrow}
For any set $\Phi$ of next-formulas, $\Phi^\downarrow = \{\beta \mid \Next\beta\in\Phi\}$. Given an elementary strict-future formula $\eta=\Ddot{\bigvee}_{i=1}^n\bigwedge_{j=1}^m\Next\beta_{i,j}$, the formula $\eta^{\downarrow}$ is defined to be $\Ddot{\bigvee}_{i=1}^n\bigwedge_{j=1}^m\beta_{i,j}$. 
\end{definition}
\begin{example}
Consider the strict-future formula $\eta = \Eventually_{[1,2]}a\;\Ddot{\Or}\;\Always_{[1,3]}b$. Then,
$
\eta^E = \Next a  \;\Ddot{\Or}\;\Next\Eventually_{[1,1]}a \;\Ddot{\Or}\; (\Next b\And\Next\Always_{[1,2]}b)
$
and
$
\eta^\downarrow = a  \Ddot{\Or}\;\Eventually_{[1,1]}a \Ddot{\Or} ( b \And \Always_{[1,2]}b)
$.
\end{example}

\vspace{-5mm}

\begin{figure}[h!]
\[
\begin{array}{lll}
(\Next) & \displaystyle{\frac{\Phi, \eta, \Next\Always\psi
}{\eta^\downarrow, \Always\psi}} 
 &  \mbox{ if }  \Phi\cup\{\eta\} \mbox{ is elementary and } \eta \mbox{ is strict-future.}
\end{array}
\]
\vspace{-5mm}
\caption{Next-state Rule}
\label{fig-next-state-rule}
\end{figure}

The {\it Next-state Rule} (in Figure \ref{fig-next-state-rule}) is applied whenever the target set of formulas is elementary. Note that, if there is not an strict-future formula $\eta$, the children of the above rule $(\Next)$ is just $\Always\psi$.

\begin{proposition}
\label{prop-downarrow}
Let $\Phi\cup\{\eta\}$ be any consistent and elementary set of formulas with strict-future formula $\eta$. Then
\begin{enumerate}[(a)]
\item 
\label{1}
For any trace $\sigma$, if $\sigma\models\Phi,\eta,\Next\Always\psi$ then $\sigma^1\models \eta^\downarrow,\Always\psi$.
\item 
\label{2}
For any (non-empty) finite trace  $\lambda = \lambda_0,\dots\lambda_{k-1}$ that is a pre-witness of $\alpha\And\Always\psi$ such that $\lambda_{k-1}\wdash\Phi$ and any $\lambda_k\in\Val(\MX\cup\MY)$, the finite trace $\lambda\cdot \lambda_k$ is a pre-witness of $\alpha\And\Always\psi$ if and only if $\lambda_k\wdash\eta^\downarrow\And\psi$.
\end{enumerate}
\end{proposition}
\begin{proof}
Both items follows from Definitions \ref{def-elementary-strict-future} and \ref{def-downarrow}, by routine application of semantic definition of $\Next$ operator.\qed
\end{proof}

According to the previously introduced set of tableau rules, we can define the set of all formulas that could appear in the construction of a tableau for a given safety specification, which is known as the closure of the given specification. 

\begin{definition} 
\label{def-closure}
We denote by $\SubForm(\beta)$ the set of all subformulas of any given
formula $\beta$. In particular,
$\SubForm(\Next^i\beta) = \{\Next^j\beta\mid 0\leq j\leq i\}
\cup\SubForm(\beta)$.  For a given safety formula $\psi$, we define
$\Variants(\psi)$ to be the union of the following four sets that
collects all the variants of subformulas $\Eventually_I$ and
$\Always_I$ that the tableau rules could introduce.
\begin{itemize}
\item[]    $
\{\Eventually_{[n,m']}\beta, \Next\Eventually_{[n,m']}\beta \mid     \Eventually_{[n,m]}\beta \in \SubForm(\psi), n\leq m'< m\}
$
\item[]
$
\{\Always_{[n,m']}\beta, \Next\Always_{[n,m']}\beta \mid \Always_{[n,m]}\beta \in \SubForm(\psi), 
n\leq m'< m\}
$
\item[]
$
\{ \SubForm(\Next^i\beta) \mid  \Eventually_{[n,m]}\beta\in \SubForm(\psi), 0\leq i\leq n \}
$
\item[]
$
\{\SubForm(\Next^i\beta) \mid  \Always_{[n,m]}\beta\in \SubForm(\psi), 0\leq i\leq n \}
$
\end{itemize}
The set $\Ordnf(\psi)$ consists of all formulas of the form 
$\Ddot{\bigvee}_{i=1}^n\bigwedge_{j = 1}^m \beta_{i,j}$ where each $\beta_{i,j}$ is in $\Variants(\psi)$.
Then, the closure of a safety specification $\varphi =\alpha\And\Always\psi$ is the finite set $\Clo(\varphi) = \Preclo(\varphi) \cup\{\Always\psi,\Next\Always\psi\}$ where
$
\Preclo(\varphi) = \SubForm(\alpha\And\psi) \cup \Variants(\psi)\cup\Ordnf(\psi).
$
\end{definition}

\subsection{A Tableau Algorithm for Realizability}

In this section we present a non-deterministic algorithm (see
Alg.~\ref{fig-algorithm}) for deciding whether a given safety
specification is realizable.
Alg.~\ref{fig-algorithm} constructs a completed tableau that analyzes
the minimal $\MX$-coverings produced by the moves (and allowed by the
input safety specification) at the successive states of the game.
\begin{algorithm}[t!]
    \caption{$\Tableau(\Phi\cup\{\chi\})$ returns $\IsOpen$: Boolean}
    \label{fig-algorithm}
\SetAlgoLined
\uIf{$\Phi$ is inconsistent}{
$\IsOpen := \False$   
}
\uElseIf{$\chi = \Always\psi$}
{
\uIf {$\Phi_0\lessdot\Phi$ for some $\Phi_0$ in the branch of $\Phi$ }{$\IsOpen := \True$}
\uElseIf{$\TNF(\Phi\And\psi)$ is not an $\MX$-covering}
{$is\_open := \Tableau(\{\cF,\Always\psi\})$;} 
\uElseIf{$\TNF(\Phi\And\psi)$ is a non-minimal $\MX$-covering}
{Let $J_1,\dots,J_m$ be all the minimal $\MX$-coverings of $\TNF(\Phi\And\psi)$;\\
$i := 0$;    
$\IsOpen:= \False$ ;\\
\While{$\Not\IsOpen \And i < m$ }{
$i := i+1$ ;\\
$\IsOpen := \Tableau(J_i\cup\{\Always\psi\})$;
}} 
\Else(\tcp*[f]{\scriptsize $\TNF(\Phi\And\psi) = \bigvee_{i=1}^n\pi_i$ is a minimal $\MX$-covering})
{
$i := 0$;    
$\IsOpen:= \True$ ;\\
\While{$\IsOpen \And i < n$ }{
$i := i+1$ ;\\
$\IsOpen :=\Tableau(\{\pi_i,\Next\Always\psi\})$;
}
}
}
\uElseIf{$\Phi = \Lambda\cup\{\eta\}$ is elementary ($\eta$ is strict-future)}{ 
$\IsOpen := \Tableau(\{\eta^\downarrow, \Always\psi\})$;
}
\Else{
$\rho := \SelectApplicableRuleTo(\Phi)$;\\
Let $1 \leq k\leq 2$ and $\Phi_1,\dots,\Phi_k$ the set of all $\rho$-children;\\
$\IsOpen := \Tableau(\Phi_1\cup\{\Next\Always\psi\})$;\\
\lIf{$k = 2\;\And\;\Not\IsOpen$}{$\IsOpen := \Tableau(\Phi_2\cup\{\Next\Always\psi\})$} 
}
\end{algorithm}
Alg.~\ref{fig-algorithm} uses recursion to explore in-depth the
branches of the tree.  The formal parameter of Alg.~\ref{fig-algorithm} is given as the union of a set of formulas $\Phi$
and a formula $\chi$ that ranges in
$\{\Always\psi,\Next\Always\psi\}$.
For deciding realizability of a safety specification
$\varphi = \alpha\And\Always\psi$, the initial call
$\Tableau(\varphi)$ is really
$\Tableau(\{\alpha\}\cup\{\Always\psi\})$.
Intuitively, Eve moves when $\chi = \Always\psi$ (at the start),
whereas Sally moves when $\chi = \Next\Always\psi$.

\begin{definition}
A branch $b$ of a tableau is any finite sequence of nodes $n_0,\dots,n_k$ such that $n_0$ is the root and $(n_i,n_{i+1})\in R$ for all $0\leq i < k-1$. If $n_k$ is a successful leaf, we say that $b$ is a successful branch. If $n_k$ is a failure leaf, we say that $b$ is a failure branch.
\end{definition}
For the sake of clarity, we omit in Alg.~\ref{fig-algorithm} the
details for loading the current branch $B$ and for performing
subsumption in nodes.  The result of Alg.~\ref{fig-algorithm} is
returned in the boolean variable $\IsOpen$.
Lines 1-5 deal with the two types of terminal nodes to which no rule
is applied,no recursive call is produced, and returns success $(\IsOpen := \True)$ or failure
$(\IsOpen := \False)$.
Line 7 produces a recursive call that immediately returns failure.
The value returned in $\IsOpen$ corresponds to whether the completed
tableau for the call parameter $\Phi\cup\{\chi\}$ is open or not (i.e
is closed).
Recursive calls in Alg.~\ref{fig-algorithm} and the notions of open and closed tableau, are related with AND-nodes, for which we introduce the following definition.
For that we next introduce the notion of bunch. 

\begin{definition}
\label{def-bunch}
Given a set of branches $H$ of a completed tableau, we say that $H$ is a  {\em bunch} if and only if for every $b\in H$ and every AND-node $n\in b$, and every $n'$ that is an $(\Always\&)$-successor of $n$, there is $b'\in H$ such that $n'\in b'$.
A completed tableau is {\em open} if and only if it contains at least one {\em bunch} such that all its branches are successful. Otherwise, when all possible bunches of a completed tableau contains a failure branch, the tableau is {\em closed}.
\end{definition}

\begin{definition}
A tableau is {\em completed} when all its branches contains a terminal node, i.e. all its branches are failure or successful.
\end{definition}

Alg.~\ref{fig-algorithm} looks for bunches of successful branches as
follows.  Lines 8-14 of Alg.~\ref{fig-algorithm} invoke a recursive
call for each minimal $\MX$-covering, according to rule
$(\Always\|)$.
When some of these calls return $\IsOpen$ for a minimal $\MX$-covering $J_i$, which is an OR-node, the iteration is finished with this result for the previous call.
The construction of the tableau for each $J_k$, by rule
$(\Always\&)$ and according to lines 15-20, produces a call for each
move $\pi_i$ in $J_k$. Moves are AND-children, hence all the calls should give $\IsOpen$ to obtain truth for $J_k$.
Finally, lines 22-23 perform the application of $(\Next)$, and lines 25-28
apply the saturation rules, when the applied rule split in two
children the second is expanded only if the first one returns that
$\IsOpen$ is false.

\begin{proposition}
Alg.~\ref{fig-algorithm} terminates and
  $\Tableau(\varphi)$ is a completed tableau.
\end{proposition}
\begin{proof}
It is easy to see that any node in $\Tableau(\varphi)$ is labelled by a finite subset of $\Clo(\varphi)$. 
Therefore, supposing that there exists an infinite branch, we could get the infinite sequence of its node labels, namely $\Phi_0,\Phi_1,\dots$. Then, every $\Phi_i\in 2^{\Clo(\varphi)}$ and for every $0\leq i < j$ it does not hold that $\Phi_i \lessdot \Phi_j$. In particular, all them must be pairwise different. Hence, we have infinitely many subsets of $\Clo(\varphi)$ that are pairwise different, which contradicts the finiteness of $\Clo(\phi)$. Therefore, any branch of $\Tableau(\varphi)$ is finite. Since the number of branches of $\Tableau(\varphi)$ is finite (by construction), termination is ensured. Lines 1-5 ensures that $\Tableau(\varphi)$ is a completed tableau.\qed
\end{proof}

Note that our algorithm assumes a procedure that given a formula calculates its $\TNF$. For the implementation, we plan to use BICA (\cite{PM15}), a Boolean simplifier for non-clausal formulas that computes a minimal (size) prime cover of the input formula.  



%% file: examples.tex
%

\section{Examples}
\label{sec:examples}

In this section, we present some representative examples that
illustrate how our tableau method works.
\begin{figure}[H]
\hspace{-3mm}
\tikzstyle{level 1}=[level distance=0.8cm, sibling distance=3cm]
\tikzstyle{level 2}=[level distance=0.8cm,sibling distance=3cm]
\tikzstyle{level 3}=[level distance=1.3cm,sibling distance=5.7cm]
\tikzstyle{level 4}=[level distance=0.8cm,sibling distance=3cm]
\tikzstyle{level 5}=[level distance=0.8cm,sibling distance=3cm]
\tikzstyle{level 6}=[level distance=0.8cm,sibling distance=3cm]
{\scriptsize
\begin{tikzpicture}
\node  (n1)  {$n_1:\;  \Always(\Next e \Iff \Next s) \qquad $}
	child {node (n10) {$n_2:\;   (\Next e \And \Next s) \Ddot{\Or} (\Next\Not  e \And \Next\Not s), \Next\Always\psi  \quad $}
      	    child {node (n1021)  {$n_3:\; (e \And  s) \Ddot{\Or} (\Not  e \And \Not s), \Always\psi \quad$}
      	            child {node (nc1)  {$n_4: \; e\And s\And (\Next e \And \Next s) \Ddot{\Or} (\Next\Not  e \And \Next\Not s), \Next\Always\psi \quad$}
      	               child {node (nc10)  {$n_6: \; e,  s,  (\Next e \And \Next s) \Ddot{\Or} (\Next\Not  e \And \Next\Not s), \Next\Always\psi \quad$}
      	                 child {node (nc11)  {$n_8: \; (e \And s) \Ddot{\Or} (\Not  e \And \Not s), \Always\psi \quad$}
      	            	}
      	            	 }
      	            }
      	            child {node (nc2)  { $\qquad n_5: \;\Not e\And\Not s\And (\Next e \And \Next s) \Ddot{\Or} (\Next\Not  e \And \Next\Not s), \Next\Always\psi \qquad$}
      	               child {node (nc20)  {$n_7: \; \Not e,  \Not s,  (\Next e \And \Next s) \Ddot{\Or} (\Next\Not  e \And \Next\Not s), \Next\Always\psi \quad$}
      	                 child {node (nc21)  {$n_9: \quad (e \And  s) \Ddot{\Or} (\Not  e \And 
      	                 \Not s), \Always\psi \quad$}
      	            	 }
      	            	}
      	            }
      		 }
    };
\draw  (n1) edge node[right]  {\tiny $(\Always\|) + (\Always\&)$} (n10); 
\path (n1021)  -- coordinate[midway] (S11) (nc1);
\path (n1021) -- coordinate[midway] (S12) (nc2);
\draw (S11) to[bend right=20]   (S12);  
\draw  (n1021) edge node {\tiny \hspace{3cm}$(\Always\|) + (\Always\&)$} (nc1);  
\draw  (n10) edge node[right]  {\tiny $(\Next)$} (n1021);
\draw  (nc10) edge node[right]  {\tiny $(\Next)$} (nc11); 
\draw  (nc20) edge node[right]  {\tiny $(\Next)$} (nc21); 
\draw  (nc2) edge node[right]  {\tiny $(\And)$} (nc20);  
\draw  (nc1) edge node[right]  {\tiny $(\And)$} (nc10);  
\draw[draw = gray]  [->] (nc11.south) to [out=190,in=-230] (n1021.north);
\draw[draw = gray]  [->] (nc21.south) to [out=-10,in=50] (n1021.north); 
\end{tikzpicture}
}
\vspace{-5mm}
\caption{A tableau that proves the realizability of $\Always(\Next e\Iff\Next s)$.}
\label{fig-tableau1}
\end{figure}
\begin{example}
\label{ex:intro-example}
  We revisit the specification $\Always(\Next e\Iff\Next s)$ discussed
  in Section~\ref{sec:intro}, for which
  $\TNF(\Next e\Iff\Next s) = (\Next p_e\And \Next s) \Or (\Next\Not
  p_e\And \Next\Not s)$, which is a minimal $\MX$-covering with just
  one move without literals.
 Hence, the only child of the root in Figure \ref{fig-tableau1} is obtained by rule $(\Always\|)$ and then $(\Always\&)$. According to Alg.~\ref{fig-algorithm}, in the node $n_3$, $\TNF((c\And s)\Or(\Not c\And\Not s))\And\psi)$ yields a minimal $\MX$-covering with two moves, hence the rule $(\Always\|)$ is applied and, after it, the rule $(\Always\&)$ yields two AND-nodes, one for each move. The first move contains $e\And s$ and the second one contains $\Not e \And\Not s$. In both branches, after saturation and application of $(\Next)$, a node already in the branch is obtained. Therefore the completed tableau has an open bunch that shows that the input is realizable.\\
Let us remark that, in practice, to decide realizability, it suffices to construct just one of the two branches in the tableau of Figure \ref{fig-tableau1}. Since the set of strict-future formulas in the label of nodes $n_4$ and $n_5$ are exactly the same, their expansion after the application of the next-state rule, in the example nodes $n_8$ and $n_9$, are also identical. 
\end{example}

\begin{example}
Let $\psi$ be the second safety formula $\eta$ of Example~\ref{ex-TNF1} where $\MX=\{p_e\}$.  In Figure~\ref{fig-tableau2} we depict the tableau for $\Always\psi$.
The $\TNF(\psi)$ (see  $\TNF(\eta)$ in Example~\ref{ex-TNF1}) is a minimal $\MX$-covering hence,  by applying rule $(\Always\|)$ and then $(\Always\&)$, nodes $n_2$ and $n_3$ are generated. According to Alg.~\ref{fig-algorithm}, they are AND-nodes, one for each move in $\TNF(\psi)$. The label of node $n_4$ is an elementary formula and the $(\Next)$-rule yields node $n_5$. Since it is not true that $\tau(n_5) \lessdot \tau(n_1)$, the tableau goes on applying rules ($\Always\|$) and ($\Always\&$) to node $n_5$.  Then,
$$
\TNF(\Always_{[0,8]} c \And \psi) = ( p_e \And c  \And \Next \Always_{[0,8]} c) \Or ( \Not p_e \And c  \And \Next \Always_{[0,8]} c)
$$
is a minimal $\MX$-covering with  two moves corresponding to nodes $n_6$ and $n_7$. As both nodes contain an elementary formula, by $(\Next)$-rule, nodes $n_8$ and $n_9$ are obtained. Finally, relations $\tau(n_8) \lessdot \tau(n_5)$ and  $\tau(n_9) \lessdot \tau(n_5)$ hold and both branches are successful.
{
\begin{figure}[H]
\hspace{1mm}

\tikzstyle{level 1}=[level distance=1cm, sibling distance=4.8cm]
\tikzstyle{level 2}=[level distance=0.8cm,sibling distance=3cm]
\tikzstyle{level 3}=[level distance=0.8cm,sibling distance=3cm]
\tikzstyle{level 4}=[level distance=1cm,sibling distance=4cm]
\tikzstyle{level 5}=[level distance=0.8cm,sibling distance=3cm]
{\scriptsize
\begin{tikzpicture}
\node  (n1)  {$n_1: \;\Always \psi = \Always(c \And (\Not p_e \Impl \Always_{[0,9]}  c) \And (\Always_{[0,9]} c \Or \Eventually_{[0,2]}\Not c)) \qquad$}
	child {node (n10) {$n_2: \; p_e \And c \And (\Next\Eventually_{[0,1]}\Not c) \Ddot{\Or} (\Next\Always_{[0,8]} c) , \Next \Always\psi \qquad$}
    }
    child {node (n11) {$n_3: \; \Not p_e \And c \And \Next\Always_{[0,8]} c, \Next \Always\psi \quad$}
      child {node (n111) {$n_4: \; \Not p_e , c , \Next\Always_{[0,8]} c, \Next \Always\psi \quad$}
      child {node (n12) {$n_5: \;\Always_{[0,8]}c, \Always\psi \quad$}
        child {node (n1223) {$n_6: \; p_e \wedge c \wedge \Next \Always_{[0,8]}  , \Next\Always\psi\quad $}
          child {node (n12231) {$n_8: \; \Always_{[0,8]} c  , \Always\psi \quad$}
            }
          }
          child {node (n1221) {$n_7: \;\Not p_e \wedge c \wedge \Next \Always_{[0,8]} c, 		\Next\Always\psi \quad$}
                child {node (n122111) {$n_9: \;\Always_{[0,8]} c, \Always\psi\quad $}
        }
      }
      }
      }
    };
    
\draw  (n11) edge node[right]  {\tiny $(\And)$} (n111);
\draw  (n111) edge node[right]  {\tiny $(\Next)$} (n12);
               
\draw  (n1221) edge node[right]  {\tiny $(\And) + (\Next)$} (n122111);
\draw  (n1223) edge node[right]  {\tiny $(\And) + (\Next)$} (n12231);

\path (n1)  -- coordinate[midway] (O11) (n10);
\path (n1) -- coordinate[midway] (O12) (n11);
\draw (O11) to[bend right=20]   (O12);  

\draw  (n1) edge node {\tiny \hspace{2.4cm}$(\Always\|) + (\Always\&)$} (n10);

\path (n12)  -- coordinate[midway] (S11) (n1223);
\path (n12) -- coordinate[midway] (S12) (n1221);
\draw (S11) to[bend right=20]   (S12);  

\draw  (n12) edge node {\tiny \hspace{2.1cm}$(\Always\|) + (\Always\&)$} (n1223); 

\draw[draw = gray]  [->] (n12231.south) to [out=190,in=-230] (n12.north); 
\draw[draw = gray]  [->] (n122111.south) to [out=-10,in=50] (n12.north);

\end{tikzpicture}
}
\vspace{-0.2cm}
\caption{Tableau for $\Always\psi$ where $\psi=c \And (\Not p_e \Impl \Always_{[0,9]}  c) \And (\Always_{[0,9]} c \Or \Eventually_{[0,2]}\Not c)$.}
\label{fig-tableau2}
\end{figure}
}

\vspace{-0.4cm} 
 
In previous Example~\ref{ex:intro-example} we noticed that one of the two branches would have not been expanded to decide that the input is realizable because the identical strict-future formulas in the label of two AND-siblings. In Figure \ref{fig-tableau2}
we do not depict the expansion of node $n_2$, because its AND-sibling $n_3$ contains a stronger strict-future formula. Note that $\Next\Always_{[0,8]} c$ is stronger than $ (\Next\Eventually_{[0,1]} \Not c \ddot{\Or} \Next\Always_{[0,8]} c)$. Then, if the expansion of node $n_3$ gets an open tableau, that should be also the case for node $n_2$. Moreover, when the strongest AND-child of a node, $n_3$ in this case, has a closed tableau, the tableau for the parent node is closed independently of the tableau for the remaining AND-children. It can be also observed that in nodes $n_6$ and $n_7$ we have the same situation as in Example \ref{ex:intro-example} both have identical strict-future formulas.
\end{example}

\begin{example}
\label{ex:four-minimal-covering}
Consider the safety specification 
$$
\Always\psi = \Always(a \Impl c , \Next p_e \Impl \Eventually_{[1,2]} a , \Next \Not p_e \Impl \Eventually_{[1,10]} \Not c ).
$$
\vspace{-0.6cm}

{
\begin{figure}[H]
\hspace{-2cm}
\tikzstyle{level 1}=[level distance=1cm, sibling distance=4.8cm]
\tikzstyle{level 2}=[level distance=1cm,sibling distance=3cm]
\tikzstyle{level 3}=[level distance=0.8cm,sibling distance=3cm]
\tikzstyle{level 4}=[level distance=0.8cm,sibling distance=3cm]
\tikzstyle{level 5}=[level distance=0.8cm,sibling distance=1.5cm]
\tikzstyle{level 6}=[level distance=1cm,sibling distance=5.5cm]
\tikzstyle{level 7}=[level distance=1cm,sibling distance=3cm]
\tikzstyle{level 8}=[level distance=0.8cm,sibling distance=3cm]
\tikzstyle{level 9}=[level distance=1cm,sibling distance=5cm]
\tikzstyle{level 10}=[level distance=0.8cm,sibling distance=3cm]
\tikzstyle{level 11}=[level distance=0.8cm,sibling distance=3cm]
{\scriptsize
\begin{tikzpicture}
\node  (n1)  {$n_1: \; \Always \psi  $}
	  child {node (n12) {$B_1: \; c \And a \And \overbrace{(\Next \Not p_e \And \Eventually_{[1,10]} \Not c ) \Ddot{\Or} 
	  ( \Next p_e \And \Eventually_{[1,2]} a )}^{\delta_1} , \Next \Always\psi\qquad$}
	      child {node (n120) {$n_2: \; c , a, ( \Next \Not p_e \And \Next \Not c ) \Ddot{\Or} (\Next \Not p_e \And \Next \Eventually_{[1,9]} \Not c ) \Ddot{\Or} ( \Next p_e \And \Next a ) \Ddot{\Or} 
	      (\Next p_e \And \Next\Eventually_{[1,1]} a ) , \Next \Always\psi\qquad $}
	       child {node (n1201) {$n_3: \; (\Not p_e  \And \Not c) \Ddot{\Or} (\Not p_e \And \Eventually_{[1,9]} \Not c ) \Ddot{\Or} ( p_e  \And  a) \Ddot{\Or} 
	      (p_e \And\Eventually_{[1,1]} a) , \Always\psi \qquad$}
	       child {node (n12011) {$n_4: \; (\Not p_e  \And \Not c) \Ddot{\Or} (\Not p_e \And \Eventually_{[1,9]} \Not c )\Ddot{\Or} ( p_e \And  a) \Ddot{\Or} 
	        (p_e\And  \Next a) , \Always\psi \qquad$}
	      child {node (nor1) {$C1,  \Always\psi \qquad$}
	      child {node (na11) {\qquad\qquad $n_5: \;  p_e \And  c \And  a \And \delta_1,  \Next \Always\psi\qquad$}
    }
	  child {node (na12) {$n_6: \;  \Not p_e \And \Not c \And \Not a \And \overbrace{(\Next \Not p_e \And \Eventually_{[1,10]} \Not c ) \Ddot{\Or} (\Next p_e \And \Eventually_{[1,2]} a )}^{\delta_1} , \Next \Always\psi\qquad$}
	      child {node (na120) {$n_7: \;  \Not p_e , \Not c , \Not a ,  (\Next \Not p_e \And \Next \Not c ) \Ddot{\Or} (\Next \Not p_e \And \Next \Eventually_{[1,9]} \Not c ) \Ddot{\Or} ( \Next p_e
	      \And \Next a ) \Ddot{\Or} ( \Next p_e
	      \And \Next\Eventually_{[1,1]} a )  , \Next \Always\psi\qquad\qquad\qquad$}
	       child {node (na1201) {$n_8: \;  (\Not p_e  \And \Not c) \Ddot{\Or} (\Not p_e
	       \And \Eventually_{[1,9]} \Not c ) \Ddot{\Or} ( p_e  \And  a) 
	       \Ddot{\Or} ( p_e  \And  \Eventually_{[1,1]} a) , \Always\psi\qquad$}
       }
       }
       }
        }
        child {node (nor2) {$C2, \Always\psi$}
        }
        child {node (nor3) {$C3, \Always\psi$}
        }
        child {node (nor4) {$C4, \Always\psi$}
        }
       }
       }%
       }
       }
    child {node (n11) {$B_2: \; \Not a \And \delta_1 , \Next \Always\psi\qquad $}
    };
\draw  (n12) edge node[right]  {\tiny $(\And) + (\Ddot\Or\Eventually\!<)$} (n120);
\draw  (n120) edge node[right]  {\tiny $(\Next)$} (n1201);
\draw  (n1201) edge node[right]  {\tiny $(\Ddot\Or\Eventually\!=)$} (n12011);

\path (nor1)  -- coordinate[midway] (P11) (na11);
\path (nor1) -- coordinate[midway] (P12) (na12);
\draw (P11) to[bend right=20]   (P12);
\draw  (n1) edge node  {\tiny \hspace{2cm}$(\Always\|)$} (n12);
\draw  (n12011) edge node  {\tiny\hspace{0.7cm}  $(\Always\|)$} (nor2);
\draw  (nor1) edge node  {\tiny \hspace{2.2cm}$(\Always\&)$} (na11);
\draw  (na12) edge node[right]  {\tiny $(\And) + (\Ddot\Or\Eventually\!<)$} (na120);
\draw  (na120) edge node[right]  {\tiny $(\Next)$} (na1201);
\draw[draw = gray]  [->] (na1201.south) to [out=190,in=-230] (n1201.north);
\end{tikzpicture}
}
\vspace{-0.5cm}
\caption{Tableau for $\Always(a \Impl c , \Next p_e \Impl \Eventually_{[1,2]} a , \Next \Not p_e \Impl \Eventually_{[1,10]} \Not c )$}
\label{fig-tableau4}
\end{figure}
}
 
 \vspace{-0.5cm}
 
The tableau in Figure~\ref{fig-tableau4} starts with 
$
\TNF(\psi)  =   (c \And a \And \delta_1) \Or  
( \Not a \And \delta_1)
$
where $\delta_1 = (\Next \Not p_e \And \Eventually_{[1,10]} \Not c ) \ddot{\Or}   (  \Next p_e \And \Eventually_{[1,2]} a )$. There are two minimal $\MX$-covering $B_1$ and $B_2$ with one move each. The tableau chooses to expand the covering $B_1$. The application of various saturation rules from $B_1$ on, in particular $(\Ddot\Or\Eventually\!<)$ and $(\Ddot\Or\Eventually\!=)$, leads to node $n_4$. Then, $(\Always\|)$ is applied to $n_4$ resulting in four OR-siblings corresponding to the four minimal $\MX$-coverings in
$\TNF(((\Not p_e  \And \Not c) \Or (\Not p_e \And \Eventually_{[1,9]} \Not c) \Or ( p_e \And  a)  \Or (p_e \And \Next a))  \And \psi)$
which are:
\begin{eqnarray*}
& \; & C1 = (p_e \And c \And a \And \delta_1) \Or (\Not p_e \And \Not c \And \Not a \And \delta_1)\\
& & C2 = (p_e \And c \And a \And \delta_1) \Or (\Not p_e \And  c \And \delta_3)\\
&  & C3 = (p_e \And \Not a \And \delta_2) \Or (\Not p_e \And \Not c \And \Not a \And \delta_1) \\
&  & C4 = (p_e \And \Not a \And \delta_2) \Or (\Not p_e \And  c  \And \delta_3)
\end{eqnarray*}
where
\begin{eqnarray*}
\delta_1 & = & (\Next \Not p_e \And \Eventually_{[1,10]} \Not c ) \Or  (\Next  p_e \And \Eventually_{[1,2]} a )\\ 
\delta_2 & = & (\Next \Not p_e \And \Eventually_{[1,10]} \Not c \And \Next a ) \ddot{\Or}  (\Next  p_e \And \Next a  )\\
\delta_3 & = &  (\Next \Not p_e \And \Eventually_{[1,9]} \Not c ) \ddot{\Or}  (\Next  p_e \And \Eventually_{[1,2]} a  \And  \Eventually_{[1,9]} \Not c )
\end{eqnarray*}
In this case, $\delta_1$ is weaker than both $\delta_2$ and $\delta_3$. We expand only  the minimal $\MX$-covering $C_1$ that contains $\delta_1$ in its two moves.
This creates two successful branches. Note that the branch starting at node $n_5$ (that we do not depicted) would be also successful. Indeed it would be finished by a node identical to $n_{8}$. Hence, the completed tableau is open.

In node $n_4$, we choose to expand the node of the minimal $\MX$-covering $C1$ with the weakest strict-future formulas. This is really convenient because if the tableau is closed for this option, we immediately know that the tableau will also be closed for any of the other minimal $\MX$-covering. Moreover, if the expansion of $C_1$ leads to an open tableau, then $n_4$ (whose children are OR-siblings) has an open tableau and there is no need to expand the remaining OR-siblings corresponding to $\MX$-coverings $C_2$, $C_3$ and $C_4$. 

It is important to highlight that it is possible to synthesise a winning strategy from any open tableau. For instance,  Figure~\ref{fig:another-good-example} represents a winning strategy by means of a finite state machine extracted from the tableau in Figure~\ref{fig-tableau4}.
\begin{figure}[H]
\centering
\begin{tikzpicture}[->,>=stealth',shorten >=2pt,auto,node distance=2.5cm,initial text=$\;$,
                    semithick]
  \tikzstyle{every state}=[fill=gray!20]
  
  \node[initial below, state] (A)       {$n_1$};
  \node[state]         (B) [right of=A] {$n_3$};

  \path 
        (A) edge[bend right, above]    node {$\small{
       a   c}$}  (B)
        
        (A) edge[bend left, above]    node {$\small{
       \Not a }$}  (B)
       
       (B) edge[loop above]   node {$\small{p_e / a c}$}  (B)
    
        (B) edge[loop below]   node {$\small{\Not p_e / \Not a \Not c}$}  (B);
    
\end{tikzpicture}
\caption{Finite state machine that represents the strategy synthesised from the open tableau in Figure~\ref{fig-tableau4}}
\label{fig:another-good-example}
\end{figure}
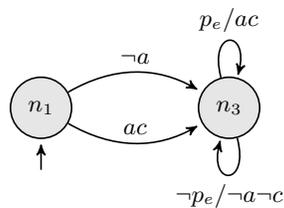
\end{example}

\begin{example}
\label{example:good}
Let $a \And \Always\psi$ be a safety specification where 
$$
\psi =  (a \Impl c) \And(p_e \Impl \Eventually_{[0,100]} \Not c) \And
(\Not p_e \Impl \Eventually_{[0,100]} a).
$$
Figure~\ref{fig:tableau5} shows the open completed tableau for $a \And \Always\psi$.
Nodes $n_2$ and $n_3$ comes from $$\TNF(a \And \psi) =  (p_e \And a \And c \And \Next\Eventually_{[0,99]}\Not  c ) \Or 
 (\Not p_e \And a  \And c )$$ 
which is a minimal $\MX$-covering.
%
In node $n_4$,
the $\TNF(\Eventually_{[0,99]} \neg c \wedge \psi)$ is
\begin{eqnarray*}
& &   (p_e \wedge c \wedge  \Next\Eventually_{[0,98]}\Not  c) \vee   (p_e \wedge \neg a  \; \wedge \; \Not  c) \vee  (p_e \wedge \neg a  \wedge c 
\wedge  \Next\Eventually_{[0,98]}\Not  c) \vee\\
& &   (\neg p_e \wedge c \wedge  a \wedge  \Next\Eventually_{[0,98]}\Not  c) \vee  (\neg p_e \wedge c \wedge \neg a \wedge \Next\Eventually_{[0,99]} a \wedge  \Next\Eventually_{[0,98]}\Not  c) \vee \\ 
& &   (\neg p_e \wedge \neg a   \wedge \neg c \wedge \Next\Eventually_{[0,99]} a) 
\end{eqnarray*}
{
\begin{figure}[H]
\hspace{-14mm}
\tikzstyle{level 1}=[level distance=1cm, sibling distance=4cm]
\tikzstyle{level 2}=[level distance=0.8cm,sibling distance=3cm]
\tikzstyle{level 3}=[level distance=1cm,sibling distance=4cm]
\tikzstyle{level 4}=[level distance=1cm,sibling distance=4cm]
\tikzstyle{level 5}=[level distance=0.8cm,sibling distance=3cm]
\tikzstyle{level 6}=[level distance=0.8cm,sibling distance=4.5cm ]
\tikzstyle{level 7}=[level distance=0.8cm,sibling distance=3cm ]
\tikzstyle{level 8}=[level distance=0.8cm,sibling distance=3cm ]
\tikzstyle{level 9}=[level distance=0.8cm,sibling distance=3cm ]
\tikzstyle{level 10}=[level distance=1.2cm]
\begin{tikzpicture}
\scriptsize
\node (n1)  {$n_1: \; a \semic \Always\psi \qquad$}
    child {node (n11) {$n_2: \; a\semic p_e \semic c \semic \Next\Eventually_{[0,99]}\Not  c \semic \Next\Always \psi \qquad$}
      				              child { node (n621) {$n_4: \;\Eventually_{[0,99]}\Not c 
      				     \semic \Always \psi \quad$}
      				     child { node (n72) {$n_5: \; C1 , \Always \psi\qquad$}
      				      child { node (n6nuevo) {$n_7: \; p_e, \neg a, \neg c, \Next \Always \psi\qquad$}
      				    child { node (n6nuevo11) {$n_9: \; \Always \psi \qquad$}
      				     }
      				              }
      				              child {node (n6nuevo2) {$n_8: \; \neg p_e, \neg a , \neg c, \Next \Eventually_{[0,99]} a
       \semic \Next\Always \psi \qquad $}
        child { node (n6nuevo22) {$n_{10}: \; \Eventually_{[0,99]} a
      				     \semic \Always \psi \qquad$}
      				     }
       }
       				}
       				child { node (n71) {$n_6: \; C_2, \Always \psi \qquad\qquad$}
      				      }
      				              }
    }
    child {node (n12) {$n_3: \; a\semic \Not p_e \semic c 
       \semic \Next\Always \psi \qquad$}
            child {node (m222) {$\Always \psi$}
           }
    };

\draw  (n1) edge  node {{\hspace{2.1cm}\tiny $(\Always\|) + (\Always\&)$}} (n11); 
\path (n1)  -- coordinate[midway] (P11) (n11);
\path (n1) -- coordinate[midway] (P12) (n12);
\draw (P11) to[bend right=30]   (P12);    
\draw (n6nuevo) edge node[right] {{\tiny $(\Next)$}} (n6nuevo11);
\draw (n6nuevo2) edge node[right] {{\tiny $(\Next)$}} (n6nuevo22);
\draw (n11) edge node[right] {{\tiny $(\Next)$}} (n621);
\draw (n621) edge node {{\hspace{2.1cm}\tiny $(\Always\|)$}} (n72);

\draw (n72) edge node {{\hspace{2.1cm}\tiny $(\Always\&)$}} (n6nuevo);
\path (n72) -- coordinate[midway] (O21) (n6nuevo);
\path (n72) -- coordinate[midway] (O22) (n6nuevo2);
\draw (O21) to[bend right=30]   (O22);

\draw[draw = gray]  [->] (m222.south) to [out=8,in=50] (n1.north); 
\draw[draw = gray]  [->] (n6nuevo22.south) to [out=-10,in=50] (n1.north); 
\draw[draw = gray]  [->] (n6nuevo11.south) to [out=190,in=-230] (n1.north); 
\end{tikzpicture}
\vspace{-0.6cm}
\caption{Tableau for 
$a \And \Always((a \Impl c) \And
 (p_e \Impl \Eventually_{[0,100]} \Not c) \And
 (\Not p_e \Impl \Eventually_{[0,100]} a)).$ }
\label{fig:tableau5}
\end{figure}
}
\vspace{-0.5cm}

The weakest option for the environment variable $p_e$ is $ (p_e \wedge \neg a  \; \wedge \; \Not  c)$ because its strict-future formula is $True$. For the environment variable $\neg p_e$ the weakest options are two:
$(\neg p_e \wedge \neg a   \wedge \neg c \wedge \Next\Eventually_{[0,99]} a)$ or $(\neg p_e \wedge c \wedge  a \wedge  \Next\Eventually_{[0,98]}\Not  c)$. So, from node $n_4$ the tableau goes on with two possible minimal $\MX$-coverings.   
\begin{eqnarray*}
& & C1 = (p_e \wedge \neg a  \; \wedge \; \Not  c) \vee (\neg p_e \wedge \neg a   \wedge \neg c \wedge \Next\Eventually_{[0,99]} a)\\
&  & C2 = (p_e \wedge \neg a  \; \wedge \; \Not  c) \vee (\neg p_e \wedge c \wedge  a \wedge  \Next\Eventually_{[0,98]}\Not  c)
\end{eqnarray*}
As explained in Example \ref{ex:four-minimal-covering}, for deciding realizability, it  suffices to expand this two $\MX$-coverings. However, since
$\MX$-covering $C_1$ gives an open tableau, we do not have to expand $C_2$. 

It is worthy to note that formulas $\Eventually_{[0,100]}\Not  c$ and $\Eventually_{[0,100]} a$ are fulfilled in at most two steps in Figure \ref{fig:tableau5}. In general, our method does not force the complete unfolding of eventualities, nor formulas of the form $\Always_{I}\eta$, unless its fulfillment check requires it. 

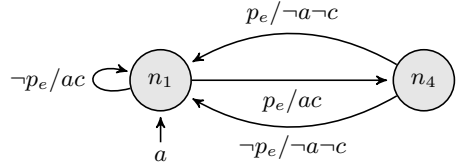
\begin{figure}[H]
\centering
\begin{tikzpicture}[->,>=stealth',shorten >=2pt,auto,node distance=3.5cm,initial text=$a$,
                    semithick]
  \tikzstyle{every state}=[fill=gray!20]
  
  \node[initial below, state] (A)       {$n_1$};
  \node[state]         (B) [right of=A] {$n_4$};

  \path (A) edge[loop left]   node {$\small{\Not p_e / a c}$}  (A)
        
        (A) edge[below]  node {$\small{p_e / a c
        }$} (B)
            
        (B) edge[bend right, above]    node {$\small{
      p_e / \Not a  \Not c}$}  (A)
        
        (B) edge[bend left, below]    node {$\small{
      \Not p_e / \Not a  \Not c}$}  (A);
    
\end{tikzpicture}
\caption{Finite state machine that represents the strategy synthesised from the open tableau in Figure~\ref{fig:tableau5}}
\label{fig:good-example}
\end{figure}

From the open tableau, we can easily synthesise a winning strategy, which is the one represented by the finite state machine of Figure~\ref{fig:good-example}.

\end{example}

\begin{example}
Let 
 $a \And \Always\psi$ be a safety specification where
$\psi =  (a \Impl c) \And
 (p_e \Impl \Next a) \And (\neg p_e \Impl \Always_{[2,10]} \Not c)$.
Figure~\ref{fig:tableau-closed} is a closed tableau that proves that $a \And \Always\psi$ is unrealizable.
To start the tableau construction, we have that
$
\TNF( a \And \psi) = (p_e \And a \And c \And \Next a ) \Or 
(\Not p_e \And a  \And c \And \Always_{[2,10]} \Not c).
$

\vspace{-1cm}
{
\begin{figure}[H]
\begin{tikzpicture}
\tikzstyle{level 1}=[level distance=1cm, sibling distance=4.8cm]
\tikzstyle{level 2}=[level distance=0.8cm,sibling distance=3.3cm]
\tikzstyle{level 4}=[level distance=1.2cm, sibling distance=5.5cm]
\tikzstyle{level 5}=[level distance=0.8cm]
\scriptsize
\node (n1)  {$n_1: \; a , \Always\psi \qquad\quad$}
    child {node (n11) {$n_2: \; p_e \And a \And c \And \Next a , \Next\Always \psi \qquad\quad$}
       child { node (m33) {$n_4: \; a,  \Always \psi \qquad\quad$}         
       }
    }
    child {node (n12) {$n_3: \; \Not p_e \And a \And c \And \Always_{[2,10]}\Not  c  , \Next\Always \psi \qquad\quad $}
     child {node (n121) {$n_5: \; \Not p_e , a ,c , \Next^2 \neg c,  \Next\Always_{[2,9]}\Not  c  , \Next\Always \psi \qquad\quad $}
            child {node (mf)  {$n_6: \; \Next \Not c , \Always_{[2,9]} \Not c, \Always \psi \qquad\quad$}
               child {node (mf1)  {$n_7: \; p_e \And c \And \Next \Not c \And \Next a \And  \Always_{[2,9]} \Not c , \Next\Always \psi \qquad\quad$}
               child {node (mfa)  {$n_9: \; p_e , c , \Next \Not c , \Next a , \Next^2 \Not c ,  \Next\Always_{[2,8]} \Not c , \Next\Always \psi \qquad\quad$}
               child {node (mf11)  {$n_{10}: \;  \Not c , a , \Next \Not c ,  \Always_{[2,8]} \Not c , \Always \psi \qquad\quad$}
               child {node (mf111)  {$n_{11}: \;  \# \qquad\quad$}
           }
           }
           }
           }
               child {node (mf2)  {$n_8: \; \Not p_e \And c \And \Next \Not c \And  \Always_{[2,10]} \Not c , \Next\Always \psi \qquad\quad$}
           }
           }
           }
    };
\draw  (n1) edge  node {{\hspace{2.5cm}\tiny $(\Always\|) + (\Always\&)$}} (n11); 
\path (n1)  -- coordinate[midway] (P11) (n11);
\path (n1) -- coordinate[midway] (P12) (n12);
\draw (P11) to[bend right=30]   (P12);   
\draw  (mf) edge  node {{\hspace{3cm}\tiny $(\Always\|) + (\Always\&)$}} (mf1); 
\path (mf)  -- coordinate[midway] (O11) (mf1);
\path (mf) -- coordinate[midway] (O12) (mf2);
\draw (O11) to[bend right=30]   (O12);    
\draw  (n11) edge node[right] {{\tiny $(\And)$}} (m33);
\draw  (n12) edge node[right] {{\tiny $(\And) + (\Always\!<)$}} (n121);
\draw  (n121) edge node[right] {{\tiny $(\Next)$}} (mf);
\draw  (mf1) edge node[right] {{\tiny $(\And) + (\Always\!<)$}} (mfa);
\draw  (mfa) edge node[right] {{\tiny $(\Next)$}} (mf11);
\draw[draw = gray]  [->] (m33.south) to [out=190,in=-230] (n1.north); 
\end{tikzpicture}
\vspace{-0.3cm}
\caption{Closed tableau for $a \And \Always((a \Impl c) \And
 (p_e \Impl \Next a) \And (\neg p_e \Impl \Always_{[2,10]} \Not c))
$.  }
\label{fig:tableau-closed}
\end{figure}
}
\vspace{-0.5cm}
The realizability result depends on the success of the AND-nodes $n_2$ and $n_3$. Once the success of node $n_2$ is ensured, the tableau goes on with the expansion of node $n_3$.  
At node $n_6$, we have that $\TNF(\Next \Not c \And \Always_{[2,9]} \Not c \And \Always \psi ) =$
\begin{eqnarray*}
& &   (p_e \And c \And \Next \Not c \And \Next a \And \Always_{[2,9]} \Not c ) \Or (p_e \And \Not a \And \Next \Not c \And \Next a \And \Always_{[2,9]} \Not c ) \Or\\ 
& &  (\Not p_e \And c \And \Next \Not c \And \Always_{[2,10]} \Not c) \Or (\Not p_e \And \Not a \And \Next \Not c \And  \Always_{[2,10]} \Not c)
\end{eqnarray*}
There are 4 possible minimal $\MX$-coverings, but it is enough to choose any of them to decide that the tableau is closed. The reason is that they have the same strict-future formula. Therefore, the tableau goes on with the AND-nodes $n_7$ and $n_8$, which correspond to the following minimal $\MX$-covering.
\begin{eqnarray*}
& & 
(p_e \And c \And \Next \Not c \And \Next a \And  \Next^2 \Not c \And \Next \Always_{[2,8]} \Not c ) \Or\\
& & 
(\Not p_e \And c \And \Next \Not c \And  \Next^2 \Not c \And \Next \Always_{[2,9]} \Not c) 
\end{eqnarray*}
As the $\TNF$ at node $n_{10}$ is $False$, this node is a failure leaf. 
This fact completes the tableau, since $n_7$ and $n_8$ are AND-siblings.   
\end{example}


%% file: games.tex
%

%% file: correctness.tex
\section{Correctness}
\label{sec:correctness}

In this section, we connect properties of the completed tableau
$\Tableau(\varphi)$ with the existence of winning strategies for Sally
or Eve in the safety formula game $\Tb(\varphi)$.
We will show that a closed tableau $\Tableau(\varphi)$ represents a
wining strategy for Eve.
However, an open tableau $\Tableau(\varphi)$ represents a wining
strategy for Sally which ensures (or proves) that the safety
specification in the root is realizable.
Indeed, the tableau construction could construct that winning strategy
and returns it as output for the user.
To carry out this proof, we first introduce tableau games and connect
them with safety games (and therefore with realizability).
Then we will connect complete tableau with the winning strategy of the
corresponding player in the tableau game.

\subsection{Safety Tableau Games}
\label{sec-games-tableaux}
In this subsection we define a class of games  $\Tb(\varphi)$ where positions are labeled by sets of formulas from the closure $\Clo(\varphi)$ of the given specification $\varphi = \alpha\And\Always\psi$.
First, we define the components of the game  $\Tb(\varphi)$. 
The set $P=P_E\cup P_S$ of positions is formed by subsets of the closure of the specification $\varphi$ where
\begin{itemize}
\item[] $P_E = \{\Phi\cup\{\Always\psi\}  \mid  \Phi\in\Preclo(\varphi) \text{ and } \Cons(\Phi)\}$
\item[] $P_S = \{\Phi\cup\{\Next\Always\psi\}  \mid  \Phi\in\Preclo(\varphi) \text{ and } \Cons(\Phi)\}$.
\end{itemize}
The set of bad positions $B\subseteq P_E$ is defined as follows: $\Phi\cup\{\Always\psi\}\in P_E$ is in $B$ whenever one of the following conditions holds
\begin{itemize}
\item[(i)]  $\TNF(\Phi\And\psi)$ is not an $\MX$-covering.
\item[(ii)] Let $J_1,\dots, J_m$ the collection of all minimal $\MX$-coverings in $\TNF(\Phi\And\psi)$. For every $1\leq k\leq m$ there exists some $i\in J_k$ such that $\Incons(\Delta)$ for all $\Delta\in\Stt(\pi_i)$.
\end{itemize}

Next, we define the moves of the arena.
In every position $p=\Phi\cup\{\Always\psi\}\in P_E$ (in particular, in the initial one where $\Phi=\alpha$) Eve, if $p\not\in B$, chooses one of the minimal $\MX$-coverings $J\subseteq I$ in $\TNF(\Phi\And\psi)=\bigvee_{i\in I}\pi_i$ and then choose some $j\in J$. The resulting position $\{\pi_j,\Next\Always\psi\}$ is in $P_S$, therefore it is the Sally's turn.
Then, Sally, in each of her turns, annotates the state with an
exhaustive application of choices for the system, where each choice is
a saturated set.
If all possible choices for Sally (according to the specification) are
inconsistent, then the position is marked as bad, and otherwise, the
position is good.
Then, Sally must get the next-state formulas from the previous
position of Sally to pass the turn to Eve.
Formally,
\begin{itemize}
\item $(\Phi\cup\{\Always\psi\},\{\pi,\Next\Always\psi\})\in T_E$  if $\pi$ is a move in $\TNF(\Phi\And\psi)$. 
\item $(\{\pi,\Next\Always\psi\}, \Delta^\downarrow\cup\{\Always\psi\})\in T_S$ if $\Delta\in\Stt(\pi)$ 
\end{itemize}

Note that the description above defines a finite state safety game,
which we call the tableau game $\Tb(\varphi)$ for specification
$\varphi$.
We will then prove the following fact.

\begin{lemma}
 $\G(\varphi)$ is winning for Sally
if and only if $\Tb(\varphi)$ is winning for Sally.
\end{lemma}

\begin{proof}
We use $\Tb$ and $\G$ as superscripts to identify the components of the two games. We first associate to each position $\Phi\in P^\Tb$ a set of positions $\Positions(\Phi)\subseteq P^\G$.
For simplicity, in the argument of $\Positions$, we will omit the formula $\Always\psi$ in every position in $P^\Tb_E$ and also $\Next\Always\psi$ in every position in $P_S^\Tb$, we just refer the set of the remaining formulas in each position. 

The set of initial states $I^\Tb = \{\{\alpha\}\}$ and $\Positions(\alpha)=\{(\epsilon,\epsilon)\}$.

For any given $\Phi\in P_E^\Tb$ we define the set of positions of its $T_S$ and $T_E$-related positions as follows.
For any move $\pi$ in $\TNF(\Phi\And\psi)$ (note that $\{\pi,\Next\Always\psi\}\in P_S^\Tb$ and $(\Phi\cup\{\Always\psi\},\{\pi,\Next\Always\psi\})\in T_E$), we define
\begin{align*}
\Positions(\pi)  = 
\{(\overline{x}\cdot v,\overline{y})\in P_S^\G\mid  (\overline{x},\overline{y})\in \Positions(\Phi) \mbox{ and } v\in\Val_\pi(\MX)\} 
\end{align*}
and for any $\Delta\in\Stt(\pi)$ such that $\Cons(\Delta)$ (note that  $\Delta^\downarrow\cup\{\Always\psi\}\in P_E^\Tb$ and $(\{\pi,\Next\Always\psi\}, \Delta^\downarrow\cup\{\Always\psi\})\in T_S$), we define
\begin{align*}
\Positions(\Delta^\downarrow)  = \{ (\overline{x}\cdot v,\overline{y}\cdot v')\in P_E^\G
\;\mid\; & (\overline{x},\overline{y})\in \Positions(\Phi), v\in\Val_\Delta(\MX)  
\mbox{ and } v'\in\Val_\Delta(\MY)\}.
\end{align*}

For any $\Phi\in P_E^\Tb$, we say that $(\overline{x},\overline{y})\in \Positions(\Phi)$ is a {\it play}, whenever $|\overline{x}| = |\overline{y}| \geq 1$ and $(x_1,y_1)\in\Positions(\{\alpha\})$. Moreover, by Proposition \ref{prop-downarrow}(\ref{b}), for any play $(\overline{x},\overline{y})\in \Positions(\Phi)$ of length $k$:
\begin{equation}
\label{fact1}
\overline{x}+\overline{y}\wdash\alpha\And\Always\psi
\text{ if and only if } x_k+y_k\wdash\Phi\And\psi
\end{equation}

Next, we prove that for all $\Phi\in P_E^\Tb$ and all play $(\overline{x},\overline{y})\in \Positions(\Phi)$: $\Phi\in B^\Tb$ if and only if $(\overline{x},\overline{y})\in B^\G$.

First, consider any $\Phi\in B^\Tb$ and any play $(\overline{x},\overline{y})\in \Positions(\Phi)$. 
 If (i) $\TNF(\Phi\And\psi)$ is not an $\MX$-covering then, by Proposition \ref{prop-non-x-covering}, there exists some $v\in\Val(\MX)$ such that for all $v'\in\Val(\MY)$,
$v + v'\not\wdash\Phi\And\psi$. Otherwise, if (ii) for every  minimal $\MX$-covering $J$ in $\TNF(\Phi\And\psi)$, there exists some $i\in J$ such that $\Incons(\Delta)$ for all $\Delta\in\Stt(\pi_i)$. Then, by Proposition \ref{prop-all-minimal-x-covering}, there exists some $v\in\Val(\MX)$ such that for all $v'\in\Val(\MY)$,
$v + v'\not\wdash\Phi\And\psi$. Therefore, by (\ref{fact1}), 
$\overline{x}\cdot v+\overline{y}\cdot v'\not\wdash \alpha\And\Always\psi$.

Second, consider any $\Phi\not\in B^\Tb$ and any play $(\overline{x},\overline{y})\in \Positions(\Phi)$. Then $\TNF(\Phi\And\psi)$ is an $\MX$-covering and there exists at least one minimal $\MX$-covering $\bigvee_{j\in J}\pi_j$ in $\TNF(\Phi\And\psi)$ such that for all $j\in J$ there exists some $\Delta_j\in\Stt(\pi_j)$ such that $\Cons(\Delta_j)$.
According to Proposition \ref{prop-Stt}, for each $j\in J$ there exists $\lambda_j\in\Val_{\Delta_j}(\MX\cup\MY)$ such that $\lambda_j\wdash\pi_j$ (hence $\lambda_j\wdash\Phi\And\psi$).
Since $J$ is an $\MX$-covering, for all  $v\in\Val(\MX)$ there exists $v'\in\Val(\MY)$ such that $ v+ v'\wdash \Phi\And\psi$. Therefore, since 
$(\overline{x},\overline{y})$ is a play, by (\ref{fact1}), we have that 
$\overline{x}\cdot v+\overline{y}\cdot v'\wdash \alpha\And\psi$
Therefore, $(\overline{x},\overline{y})\not\in B^\G$.\\
Consequently, from any winning strategy for Sally in $\G(\varphi)$ we can construct a winning strategy for Sally in $\Tb(\varphi)$ and vice versa.\qed
\end{proof}

\noindent Since $\G(\varphi)$ is winning for Sally if and only if $\varphi$ is
realizable, the following holds.

\begin{corollary}
  A safety specification $\varphi$ is realizable if and only if
  $\Tb(\varphi)$ is winning for Sally.
\end{corollary}

\subsection{Soundness and Completeness}
\label{subsec:correcntess}

Let us first recall that for every AND-node $n$ that occurs in a
branch of a bunch $H$ and every children $n'$ of $n$, there is a
branch in $H$ that includes $n'$.
On the other hand for OR-nodes there is a possible bunch for each
children.
Consequently, when a tableau $\Tableau(\Phi\cup\{\Always\psi\})$ is
closed (i.e. any possible bunch contains a failure branch) at least
one child of every AND-node produces a closed tableau.

\begin{proposition}
\label{prop-closed-tableau}
Let $\Phi$ be any consistent set of safety formulas such that \break $\TNF(\Phi\And\psi)$ is an $\MX$-covering.
If $\Tableau(\Phi\cup\{\Always\psi\})$ is closed, then
for any AND-node $\{\bigvee_{i\in I}\pi_i,\Always \psi\}$ (where $I$ is a minimal $\MX$-covering) there exists at least one child $i\in I$ such that 
$\Tableau(\{\pi_{i}, \Next\Always\psi\})$ is also closed.
\end{proposition}
\begin{proof}
Suppose that there exists an AND-node $\{\bigvee_{i\in I}\pi_i,\Always \psi\}$ such that every $\Tableau(\{\pi_{i}, \Next\Always\psi\})$ is open for all $i\in I$. Then, each $\Tableau(\{\pi_{i}, \Next\Always\psi\})$ has a bunch of successful branches, hence  $\Tableau(\Phi\cup\{\Always\psi\})$ has also a bunch of successful branches, which contradicts our hypothesis.\qed
\end{proof}

\begin{theorem}
For any given safety specification $\varphi$, the completed tableau  $\Tableau(\varphi)$ is closed if and only if the game $\Tb(\varphi)$ is winning for Eve.
\end{theorem}
\begin{proof}

To prove the left-to-right implication of the theorem, suppose that $\Tableau(\varphi)$ is closed. Therefore, any bunch of $\Tableau(\varphi)$  contains at least one failure branch.

We define a winning strategy for Eve $\rho: P_E \rightarrow P_S$ in the game $\Tb(\varphi)$. For each position $\Phi\cup\{\Always\psi\}\in P_E$, if $\TNF(\Phi\And\psi)$ is not and $\MX$-covering, then $\Phi\cup\{\Always\psi\}\in B^\Tb$ and then Eve is winning. Otherwise, by Proposition \ref{prop-closed-tableau}, for any minimal $\MX$-covering there must exists some move $\pi$ in $\TNF(\Phi\And\psi)$ such that $\Tableau(\{\pi, \Next\Always\psi\})$ is also closed.
Hence, we select an arbitrary minimal $\MX$-covering and one such move $\pi$ and define
$\rho_E(\Phi\cup\{\Always\psi\}) = \{\pi,\Next\Always\psi\}$.
Then, for every $\Delta\in\Stt(\pi)$, either $\Delta^\downarrow\cup\{\Always\psi\}\in B^\Tb$ or $\Delta^\downarrow\cup\{\Always\psi\}\in P_E$ and every
$\Tableau(\Delta^\downarrow\cup\{\Always\psi\})$ is closed. Therefore, any bunch of all these sub-tableaux has also a failure branch.\\
In order to prove that $\rho_E$ is winning for Eve, it is enough to show that for every initial play $\delta$ played according to $\rho_E$ is winning for Eve. It is easy to see that every initial play $\delta=\delta(0),\delta(1),\dots$ such that $\delta(i+1) = \rho_E(\delta(i))$ for all $\delta(i)\in P_E$ corresponds to a failure branch. Indeed, $\delta(0) =\{\alpha,\Always\psi\}$ is in $B^\Tb$ or in $P_E$, and $\Tableau(\delta(i))$ is closed for all $i\geq 0$. Then, every initial play $\delta$ played according to $\rho_E$ is finite and $\delta=\delta(0),\delta(1),\dots,\delta(k)$ for some $k\geq 0$ such that $\delta(k)\in B^\Tb$.

Conversely, suppose that $\Tableau(\varphi)$ is open, then there exists at least one bunch $H$ such that all the leaves in $H$ are successful terminal nodes in $\Tableau({\{\alpha,\Always\psi\}})$.
Note that, for every AND-node $\{\bigvee_{i\in I}\pi_i,\Always \psi\}\in H$ and every $i\in I$,
the node $\{\pi_{i}, \Next\Always\psi\}\in H$ and
$\Tableau(\{\pi_{i}, \Next\Always\psi\})$ is open.
Moreover, every branch of $H$ is ended by a node labelled by a $\Sigma$ such that $\Sigma_0\lessdot \Sigma$ for some $\Sigma_0\in H$.
Therefore, for each $\{\pi_{i}, \Next\Always\psi\}\in H$ there exists at least one $\Delta\in\Stt(\pi_i)$ such that $\Cons(\Delta)$ and one of the following two cases holds:
\begin{itemize}
    \item $\Delta^\downarrow\cup\{\Always\psi\}\in H$, or
    \item $\Phi_0\lessdot\Phi$ for some $\Phi_0\in H$ and some $\Phi\subseteq\Delta$. In this case, there also exists some $\Delta_0\in\Stt(\Phi_0)$ such that $\Delta_0^\downarrow\cup\{\Always\psi\}\in H$.
\end{itemize}
Then, according to $H$, we can define a winning strategy for Sally $\rho_S: P_S \rightarrow P_E$ in the game $\Tb(\varphi)$ as follows. 
For each position $\{\pi,\Next\Always\psi\}\in P_S$, we define
$\rho_S(\{\pi,\Next\Always\psi\}) = \Delta^\downarrow\cup\{\Always\psi\}$ for some chosen $\Delta$ such that $\Delta^\downarrow\cup\{\Always\psi\}\in H$.\\
Since $H$ and $B^\Tb$ are trivially disjoint, it is obvious that $\delta(i)\in P_S\setminus B^\Tb$ for every initial play $\delta=\delta(0),\delta(1),\dots$ played according to $\rho_S$ and all $i\in\nat$. Therefore, $\rho_S$ is winning for Sally.\qed
\end{proof}

The following follows immediately

\begin{corollary}
\label{cor-open-tableau}
A safety specification $\varphi$ is realizable if and only if
the completed tableau for $\Tableau(\varphi)$ is open. Moreover, any bunch in $\Tableau(\varphi)$ such that all its leaves are successful represents a winning strategy for Sally. 
\end{corollary}

In the following example we illustrate Corollary \ref{cor-open-tableau} providing a tableau from which the winning strategy can be extracted.

\begin{example}
\label{ex-synthesis}
We consider a variant of a synthesis problem about a simple arbiter presented in~\cite{FBS13}. The arbiter receives requests from two clients, represented by
two environment variables $\MX = \{r_1, r_2\}$, and responds by assigning
grants, represented by two system variables $\MY  = \{g_1, g_2\}$.
We specify that each request should eventually be followed
by a grant in at most three second, and that the two grants should never be assigned
simultaneously.

\begin{sidewaysfigure}
\hspace{0.5cm}
\tikzstyle{level 1}=[level distance=0.6cm, sibling distance=2.5cm]
\tikzstyle{level 2}=[level distance=1.5cm, sibling distance=4.8cm]
\tikzstyle{level 3}=[level distance=0.8cm]
\tikzstyle{level 4}=[level distance=0.6cm, sibling distance=2.5cm]
\tikzstyle{level 5}=[level distance=1.3cm,sibling distance=4.3cm]
\tikzstyle{level 6}=[level distance=0.8cm]
\tikzstyle{level 7}=[level distance=0.6cm, sibling distance=2.5cm]
\tikzstyle{level 8}=[level distance=1.3cm,sibling distance=4.3cm]
\tikzstyle{level 9}=[level distance=0.6cm, sibling distance=2.5cm]
\tikzstyle{level 11}=[level distance=1.3cm,sibling distance=4.3cm]
\tikzstyle{level 12}=[level distance=0.8cm]
\hspace{-2cm}
\begin{tikzpicture}
\scriptsize
\node (n1)  {$n_1: \; \Always\psi \quad$}
	child {node (nz) {$ C_1, \Always\psi $}
    child {node (n10) {$n_2:\;  r_1 \And \Not r_2 \And g_1 \And \Not g_2,  \Next\Always\psi $}
       child {node (n101) {$n_6: \Always\psi $}
    	}
    }
    child {node (n11) {$n_3: \; \Not r_1 \And r_2 \And \Not g_1 \And g_2,  \Next\Always\psi \qquad\qquad$}
    	child {node (n111) {$n_7: \Always\psi $}
    	}
    }
    child {node (n12) {$\qquad\qquad\qquad  n_4: \; r_1 \And r_2 \And \Not g_1 \And g_2 \And \Next\Eventually_{[0,2]}\,g_1,  \Next\Always\psi \qquad\qquad\qquad\qquad $}
           child {node (n121)  {$n_8: \; \Eventually_{[0,2]}\,g_1, \Always\psi \qquad\qquad$}
           child {node (nz2) {$ C_2, \Always\psi  $}
              child {node (n1211)  {$ m_1: \; r_1 \And \Not r_2 \And  g_1 \And \Not g_2 , 
                \Next\Always\psi \qquad\qquad\qquad\qquad$} 
                child {node (n12111) {$n_9: \Always\psi $}
    			}
            }
            child {node (n1212)  {$ m_2, m_3: \; r_2 \And g_1 \And \Not g_2 \And \Next\Eventually_{[0,2]}\,g_2, \Next\Always\psi \qquad\quad$}
              child {node (n12121)  {$n_{10}: \; \Eventually_{[0,2]}\,g_2 ,  \Always\psi \qquad\qquad$}
              child {node (nz4) {$ C_3, \Always\psi  $}
              	child {node (mf1)  {$m'_1:\; \Not r_1 \And  r_2 \And \Not g_1 \And g_2 , \Next \Always \psi \qquad\qquad\qquad\qquad$}
              		child {node (mf11) {$n_{11}: \Always\psi $}
    				}
           }
               child {node (mf2)  {$m'_2, m'_3: \; r_1  \And \Not g_1 \And g_2 \And \Next \Eventually_{[0,2]}\,g_1 , \Next \Always \psi \qquad\quad$}
           child {node (mf21)  {$ n_{12}:\;  \Eventually_{[0,2]}\,g_1 , \Always \psi $}
           }
           }
           		child {node (mf3)  {\hspace{1.5cm} \colorbox{gray!30}{$m'_4: $} $\; \Not r_1 \And \Not r_2 \And \Not g_1 \And g_2 \And \Next \Not g_2 , \Next \Always \psi $}
              		child {node (mf31)  {$ n_{13}:\; \Not g_2 ,  \Always \psi $}
              		   child {node (nz6) {$ C_4, \Always\psi  $}
              			child {node (mf311)  {$ m''_1:\; r_1 \And  \Not r_2 \And g_1 \And \Not g_2 , \Next\Always \psi  \qquad\qquad\qquad\qquad$}
              				child {node (mf3111)  {$ n_{14}:\; \Always \psi $}
           	   				}
           	   			}
           	   			child {node (mf312)  {$ m''_2, m''_3:\;  r_2 \And g_1 \And \Not g_2 \And  \Next \Eventually_{[0,2]}\,g_2,  \Next\Always \psi \qquad\quad$}
           	   				child {node (mf3121)  {$ n_{15}:\; \Eventually_{[0,2]}\,g_2, \Always \psi $}
           	   				}
           	   			}
           	   			child {node (mf313)  {$\hspace{1cm} m''_4: \; \Not r_1 \And \Not r_2 \And \Not g_2 \And \Next \Not g_2 ,  \Next\Always \psi $}
           	   				child {node (mf3131)  {$ n_{16}:\; \Not g_2, \Always \psi $}
           	   				}
           	   			}
           	   			}
           	   			child {node (nz7) {$ \;$}
           				}
           	   		}
           		}
           		}
           		child {node (nz5) {$ \;$}
           		}
         	}
           }
           child {node (n1213)  {\hspace{1.5cm}  \colorbox{gray!30}{$m_4: $} $ \; \Not r_1 \And \Not r_2 \And  g_1 \And \Not g_2 \And \Next \Not g_2,  \Next\Always\psi $} 
            }
            }
            child {node (nz3)  {$\;$}
    		}
           }
    }
    child {node (n14)  {\colorbox{gray!30}{$n_5: $} $ \; \Not r_1 \And \Not r_2 \And \Not g_1 \And \Next \Not g_2 ,   \Next\Always\psi$}
    }
   }
   child {node (nz1) {$\; $}
   };nz1

\draw  (n1) edge[right]  node {\hspace{0.3cm}\tiny $(\Always\|)$} (nz);     
\draw  (nz) edge[above]  node {\hspace{2.2cm}
\tiny $(\Always\&)$} (n11); 
\path (nz)  -- coordinate[midway] (P11) (n10);
\path (nz) -- coordinate[midway] (P12) (n14);
\draw (P11) to[bend right=10]   (P12);

\draw  (n10) edge[right]  node {\tiny $(\And) + (\Next)$} (n101); 
\draw  (n11) edge[right]  node {\tiny $(\And) + (\Next)$} (n111); 
\draw  (n12) edge[right]  node {\tiny $(\And) + (\Next)$} (n121);

\draw  (n1211) edge[right]  node {\tiny $(\And) + (\Next)$} (n12111);
\draw  (n1212) edge[right]  node {\tiny $(\And) + (\Next)$} (n12121);

\draw  (mf1) edge[right]  node {\tiny $(\And) + (\Next)$} (mf11);
\draw  (mf2) edge[right]  node {\tiny $(\And) + (\Next)$} (mf21);

\draw  (mf311) edge[right]  node {\tiny $(\And) + (\Next)$} (mf3111);
\draw  (mf312) edge[right]  node {\tiny $(\And) + (\Next)$} (mf3121);
\draw  (mf313) edge[right]  node {\tiny $(\And) + (\Next)$} (mf3131);

\draw  (n121) edge[right]  node {\hspace{0.3cm}
\tiny $(\Always\|)$} (nz2); 
\draw  (nz2) edge[above]  node {\hspace{4cm}
\tiny $(\Always\&)$} (n1211); 
\path (nz2)  -- coordinate[midway] (Q11) (n1211);
\path (nz2) -- coordinate[midway] (Q12) (n1213);
\draw (Q11) to[bend right=10]   (Q12);

\draw  (n12121) edge[right]  node {\hspace{0.3cm}
\tiny $(\Always\|)$} (nz4); 
\draw  (nz4) edge[above]  node {\hspace{3.8cm}
\tiny $(\Always\&)$} (mf1); 
\path (nz4)  -- coordinate[midway] (R11) (mf1);
\path (nz4) -- coordinate[midway] (R12) (mf3);
\draw (R11) to[bend right=10]   (R12);
      
\draw  (mf31) edge[right]  node {\hspace{0.3cm}
\tiny $(\Always\|)$} (nz6);
\draw  (nz6) edge[above]  node {\hspace{3.7cm}
\tiny $(\Always\&)$} (mf311); 
\path (nz6)  -- coordinate[midway] (S11) (mf311);
\path (nz6) -- coordinate[midway] (S12) (mf313);
\draw (S11) to[bend right=10]   (S12);
            
\draw[draw = gray]  [->] (mf11.south) to [out=190,in=-230] (n12121.north);  
\draw[draw = gray]  [->] (mf3131.south) to [out=-10,in=50] (mf31.north);  
\draw[draw = gray]  [->] (mf3121.south) to [out=-10,in=50] (n12121.north);  
\draw[draw = gray]  [->] (mf21.south) to [out=-10,in=50] (n121.north); 
\draw[draw = gray]  [->] (n101.south) to [out=190,in=-230] (n1.north); 
\draw[draw = gray]  [->] (n111.south) to [out=190,in=-230] (n1.north); 
\draw[draw = gray]  [->] (n12111.south) to [out=190,in=-230] (n121.north);  
\draw[draw = gray]  [->] (mf3111.south) to [out=190,in=-230] (mf31.north);  
\end{tikzpicture}
\caption{Open tableau for $\Always((r_1 \Impl  \Eventually_{[0,3]}\, g_1) \wedge (r_2  \Impl \Eventually_{[0,3]}\, g_2) \And $ 
$ \Not(g_1 \And g_2) \And  ((\Not r_1 \And \Not r_2) \Impl \Next \Not g_2)).$}  
\label{fig:tableau-synthesis1}
\vspace{2cm}
\end{sidewaysfigure}

We assume that initially there are not requests and impose an additional requirement to hinder the winning strategy.
The safety specification is as follows.
$$\Always\psi =  \Always((r_1 \Impl \Eventually_{[0,3]}  g_1) \wedge (r_2  \Impl \Eventually_{[0,3]}\, g_2) \And 
\Not(g_1 \And g_2) \And ((\Not r_1 \And \Not r_2) \Impl \Next\Not g_2))$$
In Figure~\ref{fig:tableau-synthesis1} we show an open  tableau for $\Always\psi$, whose construction starts with $C_1$,
the weakest minimal $\MX$-covering in $\TNF(\psi)$, composed by the labels of nodes $n_2$, $n_3$, $n_4$ and $n_5$.
At node  $n_8$, the 
weakest 
minimal $\MX$-covering in $\TNF( \Eventually_{[0,2]}\,g_1 \And  \psi )$ is $C_2$, which has the following four moves:
\begin{eqnarray*}
 & & (r_1 \And \Not r_2 \And g_1 \And \Not g_2) \Or (\Not r_1 \And r_2 \And g_1 \And \Not g_2 \And \Next\Eventually_{[0,2]}\,g_2) \Or  \\
 & & (r_1 \And r_2 \And g_1 \And \Not g_2 \And \Next\Eventually_{[0,2]}\,g_2 ) \Or (\Not r_1 \And \Not r_2 \And  g_1 \And \Not g_2 \And \Next \Not g_2)
\end{eqnarray*} 
that we name by $m_1, m_2, m_3, m_4$ in the given order. They are represented in the nodes of the tableau with the same name (in abuse of notation). 
Note that, for simplicity, we group $m_2$ and $m_3$ in the same node that omits the value of $r_1$, which is the only difference between both moves.
At node $n_{10}$, $C_3$, is  the weakest minimal $\MX$-covering in $\TNF( \Eventually_{[0,2]}\,g_2 \And  \psi )$. It  has four moves $m'_1, m'_2, m'_3, m'_4$ but $m'_2$ and $m'_3$ has been grouped.
%
And similarly, at node  $n_{13}$, where $\TNF( \neg g_2 \And  \psi )$ provides $C_4$, with the moves $m''_1, m''_2, m''_3, m''_4$.
%
Note that nodes $m'_4$, $m_4$ and $n_5$ share the same strict-future formula.
Hence, to save space, we do not depict the expansion of nodes  $m_4$ and $n_5$ since it repeats the tableau behind node $m'_4$. All in all, the completed tableau for the input specification is open.
Moreover, the winning strategy represented by the finite state machine in Figure~\ref{fig:MealyMachine} can be synthesised from this tableau.

\begin{figure}[H]
\centering
\begin{tikzpicture}[->,>=stealth',shorten >=2pt,auto,node distance=3cm,initial text=$ $,
                    semithick]
  \tikzstyle{every state}=[fill=gray!20]

  \node[initial below, state] (A)       {$n_1$};
  \node[state]         (B) [above of=A] {$n_8$};
  \node[state]         (C) [right of=B] {$n_{10}$};
  \node[state]         (D) [right of=A] {$m'_4$};  
 
  \path (A) edge[loop left]   node {$\small{\begin{aligned}
      r_1 \overline{r_2}/ g_1 \overline{g_2}\\[-0.7ex]
      \overline{r_1} r_2/ \overline{g_1} g_2 
    \end{aligned}}$}  (A)
        
        (A) edge [bend left, left] node {$\small{r_1 r_2/ \overline{g_1} g_2 
        }$} (B)
            
        (A) edge[bend right,  below]  node {$\small{\overline{r_1}\overline{r_2}/\overline{g_1}}$} (D)
        
        (B) edge[loop left]   node {$\small{
      r_1 \overline{r_2}/ g_1 \overline{g_2}}$}  (B)
        
        (B) edge[bend left,  above]  node {$ \small{r_2/ g_1  \overline{g_2} 
        }$} (C)
        
        (C) edge[bend left, below ]  node {$\small{r_1/ \overline{g_1} g_2 
        } $} (B)		
		
        (C) edge[loop right]   node {$\small{
      \overline{r_1} r_2/ \overline{g_1} g_2}$}  (C)
      
       (C) edge[bend left,  right]  node {$\small{\overline{r_1} \overline{r_2}/ \overline{g_1}  g_2 
       }$} (D)
       
        (B) edge[bend right, left]  node {$\small{\overline{r_1}\overline{r_2}/\overline{g_1}g_2}$} (D)
        
        (D) edge[bend left, left]  node {$\small{r_2/ g_1 \overline{g_2} 
        }$} (C)
        
        (D) edge[loop right]   node {$\begin{aligned}
      r_1 \overline{r_2}/ g_1 \overline{g_2}\\[-0.7ex]
      \overline{r_1} \overline{r_2}/ \overline{g_2} 
      \end{aligned}$}  (D);

\end{tikzpicture}
\caption{Finite state machine that represents the strategy synthesised from the open tableau in Figure~\ref{fig:tableau-synthesis1}}
\label{fig:MealyMachine}
\end{figure}
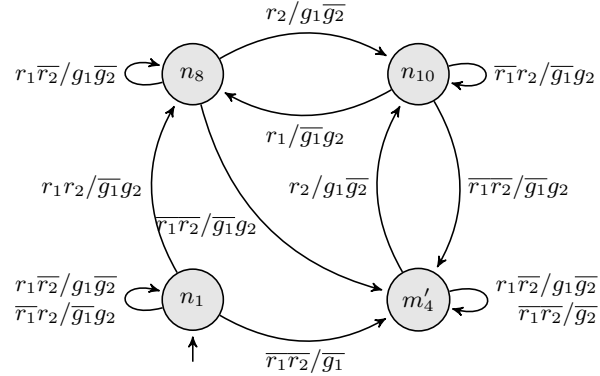
\end{example}

In Figure~\ref{fig:MealyMachine}, the machine nodes has been intentionally named as some of the tableau nodes, to make easier to see the correspondence between the transition edges labels of the form Eve/Sally and the moves in the tableau.
The forward slash separates environment variables from system variables. The negation operator is represented with a top line and the $\And$-operator has been omitted.


%% file: conclusions.tex
\section{Conclusions}
\label{sec:conclusions}

We have introduced the first tableau method to decide realizability of
temporal formulas.
Our tableau method allows to synthesize a system when the
specification is realizable.
Our realizability tableau method is based on the novel notion terse
normal form (\TNF) of formulas that is crucial in the tableau
formulation.
Our realizability tableau rules make use of the terse normal form to
precisely capture the information that each player (environment and
system) has to reveal at each step.
We have proved soundness and completeness of the proposed
method.
Our proofs imply a method to synthesize a correct system for a
realizable specification, as illustrated that in Example
\ref{ex-synthesis}.

Future work includes the implementation of the method presented in
this paper and to experiment with the resulting prototype in a
collection of benchmarks.
We also plan to extend the method to more expressive languages,
including the handling of richer propositional languages (like numeric
variables) by combining realizability tableau rules with tableau
reasoning capabilities for these domains.
This has been illustrated by the handling of enumerated types in this
paper.
Another interesting extension is a deeper analysis, including new
rules, to handle upper and lower bounds of intervals in temporal
operators, for example to accelerate a branch to reach the lower bound
$a$ of an $\Always_{[a,b]}$ operator.
We would like to ultimately extend our tableau method to richer
fragments of LTL.
